\documentclass[11pt, letterpaper]{article}

\usepackage{comment}
\usepackage{tcolorbox}
\usepackage{thm-restate}
\usepackage{subcaption}

\usepackage{hchang}

\usepackage{booktabs} % book-quality tables

\usepackage{lipsum}   % filler text
\usepackage{fancyvrb} % extended verbatim environments
  \fvset{fontsize=\normalsize}% default font size for fancy-verbatim environments
\usepackage{cleveref}

\newtheorem{theorem}{Theorem}[section]
\newtheorem{lemma}[theorem]{Lemma}
\newtheorem{claim}[theorem]{Claim}
\newtheorem{observation}[theorem]{Observation}

\newtheorem{corollary}[theorem]{Corollary}

\newcommand{\cC}{\ensuremath{\mathcal{C}}}
\newcommand{\cR}{\ensuremath{\mathcal{R}}}

\newcommand{\dist}{\ensuremath{d}}
\newcommand{\len}{\mathsf{len}}
\newtheorem{conjecture}{Conjecture}
\newtheorem{question}{Question}

\nolinenumbers

\title{Peeling Half the Onion: Embedding $k$-Outerplanar Graphs into $\ell_1$ and Trees with a Polynomial Distortion (in $k$)}
\author{
Hsien-Chih Chang\thanks{Dartmouth College, USA, \texttt{hsien-chih.chang@dartmouth.edu}. Supported by the NSF CAREER award CCF-2443017.} \and
Jonathan Conroy \thanks{Dartmouth College, USA, \texttt{jonathan.conroy.gr@dartmouth.edu}. Supported by the NSF CAREER award CCF-2443017.} \and
 William Eliot  \thanks{University of Massachusetts,  Amherst, USA, \texttt{wgeliot@gmail.com}. Supported by the NSF CAREER award CCF-2237288 and the NSF grant CCF-2517033.} \and
Hung Le\thanks{University of Massachusetts,  Amherst, USA, \texttt{hungle@cs.umass.edu}. Supported by the NSF CAREER award CCF-2237288 and the NSF grant CCF-2517033.} \and 
Vinayak\thanks{University of Massachusetts,  Amherst, USA, \texttt{vvinayak@umass.edu}. Supported by the NSF CAREER award CCF-2237288 and the NSF grant CCF-2517033.}} 
\date{}

\begin{document}

\maketitle
\begin{abstract}

Chekuri, Gupta, Newman, Rabinovich, and Sinclair~\cite{CGNRS06} showed that $k$-outerplanar graphs can be embedded into trees and $\ell_1$ with distortion $2^{O(k)}$. Their result is perhaps the strongest evidence supporting the still-open planar $\ell_1$-embedding conjecture: planar metrics can be embedded into $\ell_1$ with constant distortion. However, the exponential dependency on $k$ in their distortion bound remains the state of the art. Their embedding is obtained by a so-called \EMPH{onion-peeling approach}: removing one layer of the graph at a time at the cost of incurring $O(1)$ distortion multiplicatively, and recursively embedding the resulting $(k-1)$-outerplanar graph. 

In this paper, we devise a new onion peeling approach that peels $k/2$ layers off the input $k$-outerplanar graph at every step. As a result, we obtain the first embedding of $k$-outerplanar graphs into trees and $\ell_1$ with distortion $\poly(k)$, an exponential improvement over the best-known distortion bound \cite{CGNRS06}.  
Our result of embeddings into trees comes closer to the distortion lower bound $\Omega(k)$ for $k$-outerplanar graphs~\cite{GNRS04,CJLV08}. 
\end{abstract}

\newpage

\section{Introduction}

An \EMPH{embedding} of a metric space $(X,\delta_X)$ into another metric space $(Y,\delta_Y)$ is a function $f: X\rightarrow Y$. We say that the \EMPH{distortion} of $f$ is at most $\alpha\cdot\beta$ for parameters $\alpha, \beta \geq 1$ such that:
\begin{equation*}
    \alpha^{-1}\cdot \delta_X(x,y) \leq \delta_Y(f(x), f(y)) \leq \beta\cdot \delta_X(x,y) \quad \forall x,y \in X.
\end{equation*}
Parameters $\alpha$ and $\beta$ are called the \EMPH{contraction}%
\footnote{%
One may also consider embeddings into  \EMPH{dominating metrics}, i.e, $\alpha = 1$. For our paper, it is easier to allow contraction in the embedding.} 
and \EMPH{expansion} of $f$, respectively. 
Low-distortion metric embedding has become an essential tool in the algorithmic toolkit over the last few decades. 
Two exemplary results in this area are embeddings of general metrics into $\ell_1$ with distortion $O(\log n)$~\cite{Bourgain85} and into a distribution of trees with \emph{expected} distortion $O(\log n)$~\cite{Bartal96,Bartal98,FRT04}.
These results have countless algorithmic applications, including sparsest cut~\cite{AY98,LLR95}, network design~\cite{FRT04}, and online algorithms~\cite{AAABN06}, to name a few.  The $\Omega(\log n)$ distortion bound is optimal in both cases~\cite{LR99,LLR95,Bartal96}.
Since these classical results, a major ongoing research direction is to study low-distortion embeddings (into $\ell_1$ and trees) of metrics with additional structures; the hope is to achieve distortion $o(\log n)$, and even $O(1)$. Notable structural metrics are \emph{planar} and \emph{minor-free} metrics, which are shortest path metrics of edge-weighted planar and minor-free graphs, respectively.

\paragraph{Embedding into \boldmath{$\ell_1$}.} The Okamura-Seymour Theorem~\cite{OS81}, phrased in the language of metric embedding, states that if points of $(X,\delta_X)$ are (supported by) vertices on the outerface of a planar graph, then $(X,\delta_X)$ can be embedded isometrically (that is, with distortion $1$) into $\ell_1$. 
This basic result is a stepping stone for  Gupta,  Newman,  Rabinovich, and Sinclair~\cite{GNRS04}.
They showed that the shortest-path metrics of graphs in a class $\mathcal{G}$ are embeddable into $\ell_1$ with distortion $1$ if and only if $\mathcal{G}$ is minor-closed and does not contain $K_{2,3}$; this contains the  class of outerplanar graphs.
This simple result hints at a deep connection between topology (excluding a fixed minor) and geometry ($\ell_1$-embeddability). It is formalized via the following conjecture, which has been driving recent research on the study of planar/minor-free metrics.
\begin{conjecture}[$\ell_1$-Embedding Conjecture~\cite{GNRS04}]
\label{conjec:emb} 
Planar metric, and more generally minor-free metric, admit an embedding into $\ell_1$ with $O(1)$ distortion.
\end{conjecture}

Gupta, Sinclair, Newman, and Rabinovich~\cite{GNRS04} observed that \Cref{conjec:emb} would give an alternate characterization of minor-free graphs because the converse is true: if any graph in a class $\mathcal{G}$
 can be embedded into $\ell_1$ with $O(1)$ distortion (for any non-negative weight function on the edges), then $\mathcal{G}$ is minor-closed.  Despite intensive research since its introduction, \Cref{conjec:emb} remains widely open.   The best-known distortion for planar metrics is $O(\sqrt{\log n})$ by Rao~\cite{Rao99}, which was later extended to minor-free metrics~\cite{KLMN04,AGGNT19,AFGN22}. There have been a few different approaches~\cite{KLMN04,AFGN22,KLR19,Filtser20} toward \Cref{conjec:emb}, but none could overcome the $O(\sqrt{\log n})$ distortion barrier. On the other hand, the best-known lower bound, obtained very recently~\cite{KK25}, is $16/7\approx 2.28$. 

 An orthogonal approach is to achieve $O(1)$ distortion for well-studied subclasses of planar and minor-free graphs~\cite{GNRS04,CGNRS06,LR09,Sidiropoulos2013}. 
 Lee and Sidiropoulos \cite{LS09,LS13} showed that pathwidth-$k$ graphs are embeddable into $\ell_1$ with distortion $(4k)^{k^3+2}$. 
 (A follow-up work~\cite{AFGN22} reduced the distortion to $O(k)$ by developing a new embedding framework based on shortest-path decomposition (SPD) depth.) 
 Lee and Sidiropoulos \cite{LS09} used Robertson-Seymour decomposition~\cite{RS03}, reducing the embedding of minor-free graphs to two seemingly easier conjectures, though both of them also remain widely open: (1) the $\ell_1$-embedding conjecture of planar metrics and (2) the so-called \emph{$k$-sum conjecture}. 
 The resolution of $2$-sum conjecture~\cite{LP13} gives more credence to the $k$-sum conjecture, but nothing is known when $k\geq 3$. 
 
An evidence supporting the embedding conjecture for planar metrics is the result for $k$-outerplanar graphs by Chekuri, Gupta, Newman, Rabinovich, and  Sinclair~\cite{CGNRS06}. A graph is \EMPH{$k$-outerplanar} if its vertex set can be partitioned into $k$ consecutive layers, starting from the outerface (see  \Cref{sec:prelim} for a formal definition). The class of $k$-outerplanar graphs plays a central role in designing polynomial-time approximation schemes in planar graphs and specifically, in Baker's technique~\cite{Baker94}, which effectively reduces designing $(1+\e)$-approximation algorithms in planar graphs to those in $k$-outerplanar graphs for $k = 1/\e$.  
The authors of~\cite{CGNRS06} gave an embedding of $k$-outerplanar graphs into $\ell_1$ with distortion $2^{O(k)}$. 
One immediate corollary is that distortion $O(1)$ can be achieved for planar graphs when the number of layers $k$ is constant; 
previously, $O(1)$ distortion was only known when $k = 1$, i.e., outerplanar graphs. 
Another simple corollary is a \emph{fractional flow-cut gap} of $2^{O(k)}$ for $k$-outerplanar graphs. 
Chekuri, Shepherd, and Weibel~\cite{CSW13} strengthened this result to obtain an \emph{integer} flow-cut gap of  $2^{O(k)}$ for $k$-outerplanar graphs. 
Indeed, their flow-cut gap result holds for the more general \EMPH{$k$-shell instances} where at least one terminal of each demand pair lies in one of the $k$ outer layers. 
The $k$-shell instances can be seen as an approximate generalization of the Okamura-Seymour instances which correspond to $1$-shell instances, and fit well with an ongoing research effort to obtain generalizations~\cite{Kummar20,Kumar25}. 

All of the aforementioned results exhibit the same exponential dependence on $k$, stemming from the $\ell_1$-embedding bound of~\cite{CGNRS06}, which has stood for more than two decades. 
While the planar $\ell_1$-embedding conjecture, if true, implies that the distortion bound is $O(1)$ and independent of $k$, its refutation does not rule out a possible $O(\sqrt{\log k})$ distortion. 
Specifically, for general planar graphs, $k = O(n)$ and  the distortion becomes $O(\sqrt{\log n})$, matching the currently best known bound by Rao~\cite{Rao99}. 
In this work, we seek to narrow this huge gap between the upper and lower bounds.

\begin{question}\label{ques:embedL1_Q1}
    Can we improve the $\ell_1$-embedding distortion of $k$-outerplanar graphs substantially, to $\poly(k)$, or even to $O(\poly\log k)$? 
\end{question}

While one could ask for a distortion $O(\sqrt{\log k})$, a $\poly\log k$ distortion would be a more realistic milestone in the context of current literature.

\paragraph{Embedding into a distribution of trees.}  Let $\mathcal{D}$ be a distribution of embeddings of $(X,\delta_X)$ such that every embedding in $\mathcal{D}$ has contraction at most  $\alpha$. We say that the embedding has \EMPH{expected expansion} at most $\beta\geq 1$ if 
$\E_{f\sim \mathcal{D}} \delta_Y(f(x), f(y)) \leq \beta \cdot \delta_X(x,y)$.  
If $\mathcal{D}$ has contraction at most $\alpha$ and expected expansion at most $\beta$, we say that  $\mathcal{D}$ has {expected distortion} at most $\alpha\cdot\beta$. 
The distribution $\mathcal{D}$ induces a distribution of the host metrics $(Y,\delta_Y)$, one for each embedding in $\mathcal{D}$. 
If every metric $(Y,\delta_Y)$ is a shortest path metric of some tree, then we say that $\mathcal{D}$ is an embedding into a distribution of trees. 
Embedding into (a distribution of) trees is stronger than embedding into $\ell_1$ in the following sense: if $(X,\delta_X)$ is embeddable into trees with expected distortion $D$, then $(X,\delta_X)$ is embeddable into $\ell_1$ with (worst-case) distortion $D$. 

As in the case of $\ell_1$ embedding, it is natural to ask if the expected distortion bound $O(\log n)$ for embeddings into trees of general metrics can be improved for planar and minor-free metrics. 
In stark contrast to $\ell_1$ embeddings, the answer is overwhelmingly negative: for very simple subclasses of planar graphs, including series-parallel graphs~\cite{GNRS04,CJLV08} and grid graphs~\cite{CG04}, the expected distortion must be $\Omega(\log n)$. 
Surprisingly perhaps, the aforementioned embedding of $k$-outerplanar graphs into $\ell_1$ 
by~\cite{CGNRS06} is obtained via embedding into trees.  
While the expected distortion is $2^{O(k)}$, the fact that it does not depend on $n$ is remarkable when considering the distortion lower bound $\Omega(\log n)$ for series-parallel graphs~\cite {CG04}, a.k.a.\ planar graphs of treewidth 2. 
It is often the case that efficient algorithmic results for $k$-outerplanar graphs are obtained via a combination of planarity and small treewidth~\cite{Dorn10,DPBF10}. 
The distortion separation between treewidth-2 planar graphs and $k$-outerplanar graphs implies that $k$-outerplanar graphs enjoys additional structural properties that can potentially be utilized. 
Nevertheless, there is a glaring gap between the upper bound of $2^{O(k)}$ and the lower bound $\Omega(k)$ from known techniques~\cite{GNRS04,CJLV08}. Here we ask to narrow this exponential gap.

\begin{question}\label{ques:embedL1_Q2}
    Can we embed $k$-outerplanar graphs into a distribution of trees with expected distortion $\poly(k)$? 
\end{question}

\subsection{Our Contribution}
Our main result is a positive answer to \Cref{ques:embedL1_Q2}: an embedding $k$-outerplanar graphs into a distribution over trees with a polynomial distortion in $k$. 

\begin{theorem}\label{thm:main}
    The shortest path metric of every $k$-outerplanar graph admits an embedding into a distribution over trees with expected distortion $O(k^c)$ for some absolute constant $c$.
\end{theorem}

The value of $c$ in \Cref{thm:main} is rather large and difficult to pinpoint exactly; we use several other subroutines as black boxes, whose (large) constants are unspecified.  Since embedding into trees implies embedding into $\ell_1$ with the same distortion, we obtain an embedding into $\ell_1$ with a distortion polynomial in $k$, a significant progress toward \Cref{ques:embedL1_Q1}. 

\begin{corollary} \label{cor:l1distortion}
    The shortest path metric of every $k$-outerplanar graph admits an embedding into $\ell_1$ with distortion $O(k^c)$ where $c$ is the constant in \Cref{thm:main}.
\end{corollary}

\subsection{Technical Ideas}

\paragraph{Onion peeling: one layer at a time.} The recursive nature of the definition of $k$-outerplanar graphs suggests a simple \EMPH{onion-peeling approach} for embedding into trees: peeling off the outermost layer, recursively embedding the remaining graph into a tree, and then gluing back the tree and the outermost layer. 
More formally, let $G_k$ be the input $k$-outerplanar graph, and $\{L_1,L_2,\ldots, L_k\}$ be the partition of the vertices of $G_k$ into $k$ layers, with $L_1$ being the outerface. Let $G_{k-1} = G_{k}\setminus L_1$,  the graph obtained from $G_{k}$ by removing the outerface. Observe that $V(G_{k-1}) = L_2\cup \ldots \cup L_k$ and hence $G_{k-1}$ is $(k-1)$-outerplanar. We recursively embed each connected component of $G_{k-1}$ into a (random) tree. Let $F_{k-1}$ be the resulting forest, one tree from the embedding of each component of $G_{k-1}$. Let $H$ be the graph obtained by gluing each tree of $F_{k-1}$ into a face of $L_1$, following the drawing of $G_{k-1}$ on the plane. A key observation is that each biconnected component of $H$ is a Halin graph, a planar graph obtained by connecting the leaves of a tree to a cycle. From a structural standpoint, the Halin graph is only slightly more complicated than outerplanar graphs. Indeed, the authors in~\cite{CGNRS06} provide a reduction to outerplanar graphs by applying ideas from the Steiner point removal problem~\cite{Gutap01}, breaking the trees inside each face into smaller subtrees, and contracting them into the vertices on $L_1$. The resulting graph, denoted by $H^-$, is now an outerplanar graph, and one can invoke a known embedding of an outerplanar graph into trees with $O(1)$ distortion~\cite{GNRS04}.

For the distortion analysis, let $s_k$ be the distortion of the embedding of $G_{k}$. Since the embedding of $H$ into a tree has distortion $O(1)$, we have:
\begin{align}\label{eq:distorion}
    s_k = O(1) \cdot s_{k-1}   \qquad \text{and} \qquad s_1 = O(1)
\end{align}
The base case $G_1$ is an outerplanar graph and hence $s_1 = O(1)$~\cite{GNRS04}. Solving \Cref{eq:distorion} gives $s_k = O(1)^{k} = 2^{O(k)}$. The constant in the base of the exponent is unspecified in~\cite{CGNRS06} as it is not simple to calculate. We remark that the resulting tree in the embedding of $G_k$ is non-spanning due to the Steiner point removal step.

Stronger versions of this embedding, including embedding into spanning trees~\cite{Emek11} and integer flow-cut gap~\cite{CSW13}, follow the same onion peeling approach, and therefore, incur the same exponential dependence on $k$.

\paragraph{What do we do? Peel half the onion.}  
To motivate our technique, let us describe a (faulty) approach to reduce the distortion to $\poly(k)$. 
First, we divide $V(G_k)$ into two sets $V_{1} =  L_1\cup \ldots \cup L_{\floor{k/2}}$ and $V_2 =  L_{\floor{k/2}+1}\cup \ldots \cup L_{k}$. Let $H_1 = G_k[V_1]$ and $H_2 = G_k[V_2]$. Since $H_1$ and $H_2$ are subgraphs induced by (roughly) $k/2$ layers, they are $(k/2)$-outerplanar. Next, we recursively embed $H_1$ and $H_2$ into two trees\footnote{Graph $H_2$ might not be connected and hence $T_2$ is a forest instead of a tree. To simplify the discussion, we assume that $H_2$ is connected.}, say $T_1$ and $T_2$, respectively. Let $K$ be the graph obtained by adding edges between $V_1$ and $V_2$ to $T_1\cup T_2$. That is, $V(K) = V(G_k)$ and $E(K) = E(T_1)\cup E(T_2)\cup \left(E(G_k)\setminus E(H_1\cup H_2)\right)$. 
If---a big IF---one can embed $K$ into a tree $T$ with expected $O(1)$ distortion, then $T$ is an embedding of $G_k$ with a polynomial expected distortion.  To see this, let $s_k$ be the distortion of $T$. Since  $H_1$ and $H_2$ are $(k/2)$-outerplanar, the maximum distortion  in $T_1\cup T_2$ of edges of $G_k$ inside $H_1$ and $H_2$ is at most $s_{k/2}$.  Since $T$ has  distortion $O(1)$ in  $K$, the distortion of the edges in $E(G_k)\setminus E(H_1\cup H_2)$, which suffer the maximum distortion over all edges of $G_k$, is at most:
\begin{equation}\label{eq:dist-half}
    s_k = O(1)\cdot s_{k/2}.
\end{equation}
Given $s_k = k^{O(1)}$, this is a polynomial bound. 

As alluded to above, the problem is that it is unclear how to embed $K$ into trees with expected distortion $O(1)$. The outerplanarity of $K$ could still be $\approx k$, and there is no clear evidence that embedding $K$ would be easier than embedding $G_k$. Nevertheless, an important takeaway of this faulty approach is that we can improve the distortion substantially by removing $k/2$ layers at a time; essentially, we \EMPH{peel half the onion}.

A major challenge with peeling half the onion is to combine the tree embeddings of the inner and outer components into a structure that has outerplanarity less than $k$. To tackle this challenge, we perform tree embedding only on the inner $(k/2)$-outerplanar components, denoted by $H_2$. In each component, we ``empty'' those vertices that are not terminals (which eventually hang off the terminals as trees) to obtain an outerplanar embedding. As a whole, this procedure reduces the outerplanarity to $\lfloor k/2\rfloor$. One then recurses.

In more detail, let $T$ be the vertices on layer $L_{\lfloor k/2\rfloor}$, called \EMPH{terminals}. We ``break open'' each component of $H_2\setminus T$ by applying the following sequence of operations. 
First, we apply a \emph{low-diameter decomposition} at various scales, based on distance to $T$: for every scale $i$, we cluster the vertices which are at distance around $\Theta(2^i)$ from $T$ into clusters of diameter $\Theta(2^i)$.
Next, we introduce a \emph{Steiner vertex} for each cluster and connect it to its corresponding cluster by adding an edge to some arbitrary ``root'' vertex on the outermost layer of the cluster. Then we add an edge between two Steiner points if there is an edge between the corresponding clusters. The weight of the edge is proportional to the distance between the clusters as well as the distance to the terminal layer.
Morally speaking, one should imagine that the edges incident to a Steiner vertex representing a scale-$i$ cluster have weight $\Theta(2^i)$. 
We call the resulting graph the {Steiner graph $G_S$}.
Finally, we apply the \emph{Steiner point removal} (SPR) technique of~\cite{chang2024shortcut} with the terminal set $T$ to obtain an outerplanar embedding with clusters hanging off the terminal vertices. Call this graph \smash{$\hat G$}. The clusters attached to vertices in $T$ are, however, $(k/2)$-outerplanar components rather than trees. Since each attached cluster is itself a $(k/2)$-outerplanar component, they can be recursively embedded independently into trees. As a result, we get a $(k/2)$-outerplanar embedding of $G$, and we can recurse further.

Conceptually, our approach is reminiscent of the technique \cite{CGNRS06} used for embedding Halin graphs into outerplanar graphs. 
In both cases, we carve out clusters by maintaining some notion of preserving distance to terminal cycle, and then perform a Steiner point removal procedure. The key difference is that our procedure carves out clusters and performs SPR in a general planar graph, rather than a Halin graph. As we describe below, this causes some subtle technical challenges that we overcome with a more careful clustering method.

\paragraph{A key subtlety in the distortion analysis.} Let's analyze the distortion incurred at each iteration. We aim for $O(1)$ distortion. (Indeed, each edge participates in the cutting procedure in at most $\lceil\log_2 k\rceil$ levels, so the expected distortion over all iterations will remain $\poly(k)$). Consider an edge $e = (u,v)$ in $G$ at scale $i$, meaning that the endpoints of $e$ are at distance $\Theta(2^i)$ from $T$ in $G$.
The edge $e$ is cut with probability at most $w(e)/\Theta(2^i)$; if it is cut, then the construction of $G_S$ guarantees that there is a path between the endpoints of $e$ with length $\Theta(2^i)$. Thus, $\mathbb{E}[\dist_{G_S}(u,v)] \le O(1) \cdot \dist_G(u,v)$. Now, what is the distortion incurred when we apply Steiner point removal to $G_S$ to produce \smash{$\hat G$}?
Consider an edge $e_S = (u_S, v_S)$ in $G_S$. If the endpoints of $e_S$ are terminals, then SPR guarantees that $\dist_{\hat G}(u_S, v_S) \le O(1) \cdot \dist_{G_S}(u_S, v_S)$. The difficult case is when the endpoints of $e_S$ are not terminals. Suppose that the endpoints of $e_S$ represent scale-$i$ clusters. In this case we have $\dist_{G_S}(u_S, v_S) = \Theta(2^i)$.  One can show that the SPR procedure produces distortion proportional to the distance between $e_S$ and $T$ \ul{in $G_S$}: that is, $\dist_{\hat G}(u_S, v_S) \le \Theta(2^i) + O(\dist_{G_S}(u_S, T) + \dist_{G_S}(v_S, T))$. 
So, to prove that $\dist_{\hat G}(u_S, v_S) = \Theta(2^i)$, we must show that $\dist_{G_S}(u_S, T) = O(2^i)$ and $\dist_{G_S}(v_S, T) = O(2^i)$.
One may now hope to make the following argument:
\begin{quote}
    Let $u$ and $v$ be arbitrary vertices in $V(G)$ that are contained in the clusters $u_S$ and $v_S$, respectively. By construction of $G_S$, we have that $\dist_{G_S}(u_S, T) \approx \dist_{G_S}(u, T)$ and $\dist_{G_S}(v_S, T) \approx \dist_{G_S}(v, T)$. 
    
    Because $u_S$ and $v_S$ are scale-$i$ clusters, we have $\dist_G(u,T) = O(2^i)$ and $\dist_G(v,T) = O(2^i)$. The expected distortion guarantee of $G_S$ implies that $\mathbb{E}[\dist_{G_S}(u,T)] \approx \mathbb{E}[\dist_{G_S}(v,T)] = O(2^i)$.
    So we conclude that $\mathbb{E}[\dist_{\hat G}(u_S, v_S)] = O(2^i)$ as desired.
\end{quote}
But there is a flaw in this analysis: a bound on the expected distance $\mathbb{E}[\dist_{G_S}(v,T)]$ is not enough for us!
If we try to make this analysis precise, we find that we need to bound the \emph{conditional expectation} of $\mathbb{E}[\dist_{G_S}(v,T)]$, conditioned on a certain edge $(u,v)$ in $G$ being cut by the random clustering (or equivalently, conditioned on a certain edge $(u_S, v_S)$ appearing in the graph $G_S$). We do not know how to bound the conditional expectation if we use the standard low-diameter decomposition to cluster.

Our fix is to use a \EMPH{stochastic shortcut partition} \cite{CCCLPP25} when clustering. With this partition, we will be able to \emph{deterministically} guarantee that, for any vertex $v$ in $G$ in a scale-$i$ cluster, we have 
$\dist_{G_S}(v,T) = O(1) \cdot \dist_G(v,T) = \Theta(2^i)$.
(Note that this guarantee may be, a prior, somewhat surprising---clearly it is impossible to deterministically preserve \emph{all} distances by $O(1)$ factor when embedding into a tree, but here we show that it is possible to deterministically preserve the \emph{vertex-to-$T$} distances.)
Indeed, the shortcut partition will let us find a path in $G$ from $v$ that walks through only $O(1)$ clusters before reaching a scale-$(i-1)$ cluster; by repeatedly finding such a path, we obtain a path in $G$ from $v$ to some terminal vertex in $T$, which intersects only $O(1)$ scale-$j$ clusters for every $j \le i$. A geometric series implies that, in the Steiner graph $G_S$, the corresponding path has length $\sum_{j \le i} \Theta(2^j) = \Theta(2^i)$. We conclude that $\dist_{G_S}(v,T) = \Theta(2^i)$ as desired.

\begin{figure}[ht]
    \centering
    \includegraphics[width=0.66\linewidth]{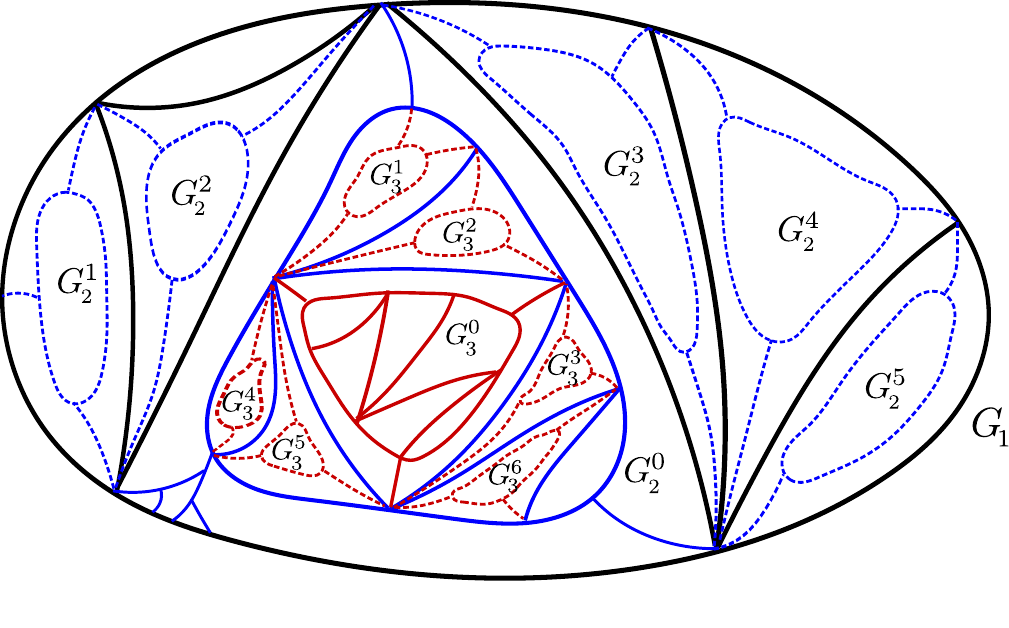}
    \caption{An illustration of the layer decomposition in a $3$-outerplanar graph. The black, blue, and red vertices are part of the layers $L_1$, $L_2$, and $L_3$ respectively. }
    \label{fig:3op-layer-decomp}
\end{figure}

\subsection{Notations and Preliminaries}\label{sec:prelim}

Let $G = (V, E, w)$ be a plane graph (that is, a planar graph together with a drawing in the plane). We denote the set of vertices on the $i^{\text{th}}$ \EMPH{outerplanar layer} as \EMPH{$L_i$}. More precisely, let $L_1$ denote the set of vertices on the outer face of $G$; for $i > 1$, let $L_i$ denote the set of vertices on the outer face of $G[V \setminus \bigcup_{j < i} L_j]$.
We say $G$ is \EMPH{$k$-outerplanar} if all vertices of $G$ are contained in $k$ outerplanar layers, $\bigcup_{i\in[k]} L_i$.

We now denote the subgraphs $\{G_i^j\}$ for $j=1,2\ldots$, where each $G_i^j$ is a connected component in the graph $G[V\setminus \bigcup_{l<i} L_l]$, i.e. the graph obtained by removing the outer $i$ layers of $G$. 

For any two vertices $u,v \in V$, we let $d_G(u,v)$ denote the distance between $u$ and $v$ according to the shortest path metric in $G$. For an edge $e = (u,v)$, we assume without loss of generality that its weight is set as $w_G(u, v) = d_G(u, v)$. The distance from a vertex $v$ to a set $X \subseteq V$ is $d_G(v, X) = \min_{x\in X} d_G(v, x)$.

\section{The Big Picture}

The following recursive procedure, $\EMPH{\textsc{TreeEmbed}}$, outputs the tree embedding of input weighted $k$-outerplanar graph $G$. 
The idea is to run our novel procedure $\EMPH{\textsc{EmbedKOuterplanar}}$ on each individual component of the subgraph induced by the inner $k/2$ layers, which turns each component into an outerplanar graph $\hat H$ with $(k/2)$-outerplanar components hanging off the outermost face of $\hat H$. 
Since these subgraphs are separated from the rest of the graph via single edge connections, we can treat them separately, and therefore we can recurse on each of these components, without accumulating distortion on any other edge of the graph. 
One round of this procedure reduces the outerplanarity by a multiplicative factor of two. 
For the base case when the graph is at most 2-outerplanar, the procedure simply calls the black box algorithm of tree embedding of 2-outerplanar graphs, 
$\textsc{2-OuterplanarToTree}$ procedure of~\cite{gnrs99,CGNRS06}, which gives a constant factor distortion embedding of $2$-outerplanar graphs into random trees.

\begin{figure}[h!]
\small
\centering
\begin{algorithm}
\textul{$\textsc{TreeEmbed}(G)$}: \+
\\ if $G$ is at most $2$-outerplanar: \+
\\      return $\textsc{2-OuterplanarToTree}(G)$ \-
\\  $G_{\lfloor k/2\rfloor}^1, \dots, G_{\lfloor k/2\rfloor}^L \gets$ $\lceil k/2\rceil$-outerplanar components of $G$
\\ $T^j \gets$ the outerface of $G_{\lfloor k/2\rfloor}^j$ for each $j\in[L]$ 
\\  for $i$ from $1$ to $L$: \+
\\      $(\hat{H}, \mathcal{C}) \gets$ \textsc{EmbedKOuterplanar}$(G_{\lfloor k/2\rfloor}^i, T^i)$
\\      for each $C \in \mathcal{C}$: \+
\\          $T_C \gets \textsc{TreeEmbed}(C)$ \-\-
\\ $\bar{G} \gets (\lfloor k/2 \rfloor+1)$-outerplanar graph formed by replacing $G^j_{\lfloor k/2\rfloor}$ with $\hat{H}^j \setminus \mathcal{C}^j$ for each $j\in [L]$ 
\\ $T' \gets \textsc{TreeEmbed}(\bar{G})$
\\ return the tree formed by gluing $\{T_C\}_{C\in \mathcal{C}}$ to $T'$ via their cut-vertices \-

\\
\\
\textul{$\textsc{EmbedKOuterplanar}(H,T)$}: \+
\\  $(H_S,\mathcal{C}) \gets \textsc{LowDiameterDecomposition}(H, T)$
\\  $\hat{H} \gets \textsc{SteinerPointRemoval}(H_S, \cC, T)$
\\  return $(\hat{H}, \mathcal{C})$\-
\\
\\
\textul{$\textsc{SteinerPointRemoval}(H_S, \cC, T)$}: \+
\\  $\hat{H} \gets \textsc{PlanarSPR}(H_S\setminus \cC, T)$
\\  for each $C\in\cC$:\+
\\      Attach $C$ to associated terminal $t\in T$ in $\hat{H}$.\-
\\  return $(\hat{H}, \cC)$
\end{algorithm}
\caption{The procedures \textsc{TreeEmbed}, \textsc{EmbedKOuterplanar}, and \textsc{SteinerPointRemoval}.}
\label{fig:tree-embed}
\end{figure}

\begin{figure}[t]
    \centering
    \begin{subfigure}{0.3\textwidth}
        \includegraphics[width=0.95\textwidth]{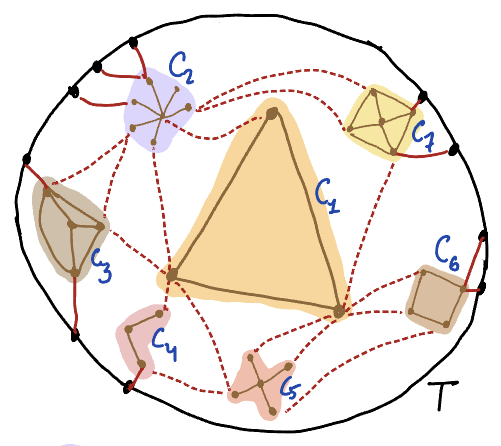}
         \caption{}
    \end{subfigure}
    \begin{subfigure}{0.3\textwidth}
        \includegraphics[width=1.0\linewidth]{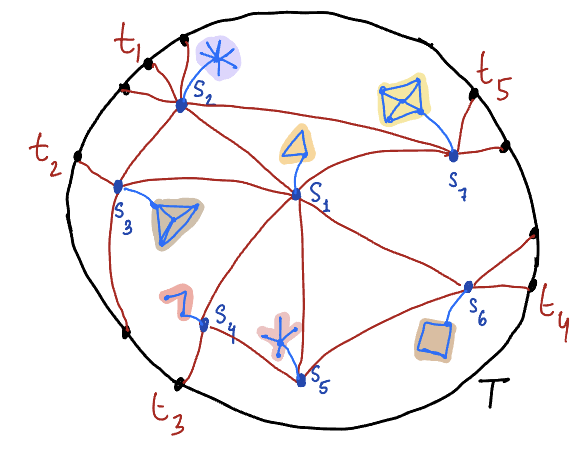}
        \caption{}
    \end{subfigure}
    \begin{subfigure}{0.3\textwidth}
        \includegraphics[width=1.0\linewidth]{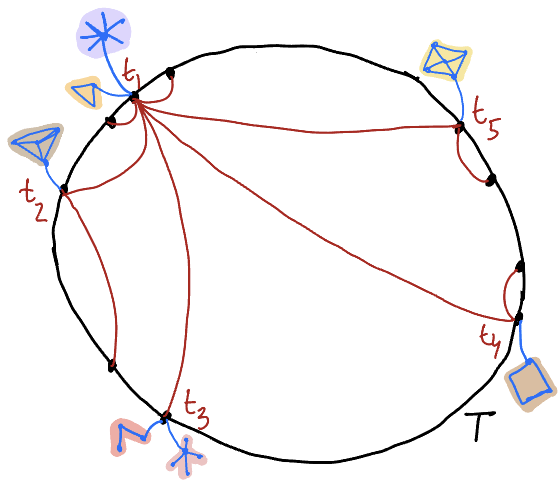}
         \caption{}
    \end{subfigure}
    \caption{One iteration of the $\textsc{EmbedKOuterplanar}$ procedure, with terminal set $T$. Figure (a) depicts the original graph $G$ with clusters after \textsc{LowDiameterDecomposition}. Figure (b) depicts the graph $G_S$, with added Steiner vertices. Figure (c) depicts the graph $\hat{G}$, after \textsc{SteinerPointRemoval} on $G_S$.}
    \label{fig:algorithm-stages}
\end{figure}

\medskip
The procedure $\textsc{EmbedKOuterplanar}$ is where the heart of the algorithm lies. $\textsc{EmbedKOuterplanar}$ calls two subroutines: 
\begin{enumerate}
\item \textsc{LowDiameterDecomposition}$(H, T)$: on input of a graph $H$, the subroutine returns a new $(k/2)$-outerplanar graph $H_S$ called the \emph{Steiner graph}, 
\item $\textsc{SteinerPointRemoval}(H_S, \cC, T)$: on input of a graph $H_S$ with outerface $T$, the subroutine returns another graph $\hat{G}$ with the same outerface $T$ but with all Steiner vertices $S$ removed, and every cluster $C_i\in \cC$ is now connected to $T$ directly by a single edge; the expected distortion of $\hat{G}$ is at most a constant factor of $H_S$. 
\end{enumerate}
We now state the distortion guarantees of our main procedure, $\textsc{TreeEmbed}$.
\begin{theorem}
\label{thm:tree-embed}
    For any edge-weighted $k$-outerplanar graph $G$,
    the procedure $\textsc{TreeEmbed}(G)$ returns a distribution $\mathcal{D}$ over tree embeddings of $G$, such that for every pair of vertices $(u,v)$ in $G$ and tree embedding $\hat G$ in the support of $\mathcal{D}$, we have:
    \[
    d_G(u,v) \leq d_{\hat G}(u, v)\quad \text{and}\quad \E_\mathcal{D}[d_{\hat G}(u, v)] \leq f(k)\cdot d_G(u, v),
    \] 
    where $f(k) = c^{\log_2 k}$ for some constant $c>1$.
\end{theorem}
We summarize the guarantee of $\textsc{EmbedKOuterplanar}(H, T)$ in the following lemma.

\begin{lemma}
\label{lem:embed-outerplanar}
    Let $H$ be a $(k/2)$-outerplanar graph, and let $T$ denote the outer face of $H$.
    There is a stochastic procedure $\textsc{EmbedKOuterplanar}(H, T)$ that returns a plane graph $\hat H$ with $V(H) = V(\hat H)$ and some $(k/2)$-outerplanar clusters $\mathcal{C}$ such that:
    \begin{itemize}
        \item $\hat H$ comprises the face $T$, some chords between vertices of $T$, and each $(k/2)$-outerplanar cluster in $\mathcal{C}$ is attached to some vertex of $T$ by a single edge.
        \item For every pair of vertices $(u,v)$ in $\hat{H}$, we have $d_{H}(u,v) \le d_{\hat H}(u,v)$ and $\mathbb{E}[d_{\hat H}(u,v)] \le O(1) \cdot d_H(u,v)$.
    \end{itemize}
\end{lemma}

\begin{observation}
\label{lem:edge-sep}
    Let graph $G$ contain a cut edge whose removal results in graphs $G_1$ and $G_2$. Suppose $G_1$ (resp.\ $G_2$) probabilistically embed into tree metric with distortion $\alpha_1$ (resp.\ $\alpha_2$). Then, $G$ can be probabilistically embedded into tree metric with distortion $\max\{\alpha_1, \alpha_2\}$. 
\end{observation}
\begin{proof}
    Let $G_1$ (resp.\ $G_2$) probabilistically embed into tree $T_1$ (resp.\ $T_2$). Take the tree $T$ formed by gluing together $T_1$ and $T_2$ via the unique cut edge. For any edge $(u,v)$ in $G$, the shortest path in $T$ is entirely contained in either $T_1$ or $T_2$ (except for the cut edge, which embeds isometrically). Therefore, the multiplicative distortion faced is $\max\{\alpha_1, \alpha_2\}$.
\end{proof}

Now we are ready to prove the correctness of the $\textsc{TreeEmbed}$ algorithm.
\begin{proof}[of~\Cref{thm:tree-embed}]
    We proceed by induction on $k$. We are going to prove the distortion guarantee to be $O(c^{\log_2 k})$ for some large constant $c$. 
    The base case is when $k\le 2$, i.e. $G$ is at most a 2-outerplanar graph. 
    The procedure $\textsc{2-OuterplanarToTree}$ (of~\cite{gnrs99,CGNRS06}) performs stochastic embedding of $2$-outerplanar graphs into tree with constant distortion.

    For general $k >2$: let \smash{$G_{\lfloor k/2\rfloor}^1, G_{\lfloor k/2\rfloor}^2, \dots, G_{\lfloor k/2\rfloor}^L$} be the $\lceil k/2\rceil$-outerplanar components of $G$, as defined previously, and let \smash{$T^j = G_{\lfloor k/2\rfloor}^j\cap L_{\lfloor k/2\rfloor}$} denote the corresponding terminal layers. Let $G_\text{out}$ be the graph obtained by removing the layers $L_{\lfloor k/2\rfloor},\ldots, L_{k-1}$ from $G$.
    Let $H^j, \cC^j$ be the resulting graph and set of clusters respectively, returned by the \smash{$\textsc{EmbedKOuterplanar}(G_{\lfloor k/2\rfloor}^j, T^j)$} procedure for each $j\in [L]$, and consider $\bar H^j = H^j\setminus \cC^j$. Let $H$ be the resulting graph formed by gluing together $G_\text{out}$ with $\bigcup H^j$.

    Fix $u$ and $v$ in $G$, and let $P$ be a shortest path from $u$ to $v$ in $G$.  
    Split $P$ into maximal subpaths $P_1, \dots, P_s$ such that each $P_i$ lies either outside all $H^j$s, or entirely in one $H^j$. 
    Denote $x_i$ and $y_i$ the two endpoints of $P_i$.
    If subpath $P_i$ lies outside all $H^j$s, we have $d_H(x_i,y_i) \leq \len_G(P_i)$, since $G[H\setminus \cup H^j]=G_\text{out}$. 
    If subpath $P_i$ lies in $H^j$, we have $\mathbb{E}[d_H(x_i, y_i)]\leq O(1)\cdot d_{H^j}(x_i,y_i) \leq O(1)\cdot \len_G(P_i)$. 
    Therefore,
    \[
    \mathbb{E}[d_H(u,v)]\leq \sum_{i=1}^s d_H(x_i, y_i)\leq \sum_{i=1}^s c\cdot \len_G(P_i) = c\cdot d_G(u,v).
    \]
    To prove non-contraction, we take the $u$--$v$ shortest path $Q$ in $H$, and decompose it into sub-paths $Q_1, Q_2,\ldots$ such that each $Q_i$ lies either completely inside $G_\text{out}$ or in some $H^j$. If $Q_i(x_i, y_i)$ lies in $G_\text{out}$, then $d_H(x_i, y_i)\geq d_G(x_i, y_i)$. Otherwise, if $Q_i$ lies in $H^j$ (for some $j$), then $d_H(x_i, y_i) \ge d_{H^j}(x_i, y_i)\ge d_{G^j_{\lfloor k/2 \rfloor}}(x_i, y_i)\ge d_G(x_i, y_i)$. Therefore, 
    \[
    d_H(u, v) = \sum_i d_H(x_i, y_i) \geq \sum_i d_G(x_i, y_i) \geq d_G(u, v). 
    \]
    
    We first analyze edges in $H$ contained entirely in some cluster $C\in\cC$. For each $\lfloor k/2\rfloor$-outerplanar cluster $C\in \cC$, we run the procedure $\textsc{TreeEmbed}(C)$, and let $G'$ be the resulting graph. Let $\{T_C\}_{C\in\cC}$ be the set of trees corresponding to the clusters $C\in \cC$. For each $(u,v) \in E(T_C)$, by induction hypothesis,
    \begin{align*}
        d_H(u, v) \leq d_{T_C}(u, v) \quad \text{and} \quad \E[d_{T_C}(u, v)] \leq c^{\log_2 \lfloor k/2 \rfloor} \cdot d_H(u, v),
    \end{align*}
    which, in turn, implies that $d_{G}(u, v) \leq d_{T_C}(u, v)$ and $\mathbb{E}[d_{T_C}(u, v)] \leq c^{\log_2 k}\cdot d_{G}(u, v)$. Moreover, each cluster $C\in\bigcup \cC^j$ is now embedded into a tree. The graph $G'$ is $(\lfloor k/2\rfloor + 1)$-outerplanar. 
    This is easy to see, since the layers $\bar L_i$ in the layer decomposition of $H$ are exactly the layers $L_i$ in $H^j$ for $1\leq i\leq \lfloor k/2\rfloor$. Moreover, every vertex in the set $\bigcup_{j=\lfloor k/2\rfloor+1}^{k}L_{j}$ is now part of the layer $L_{\lfloor k/2 \rfloor}$. 

    We now look at the graph $\bar G = G' \setminus \bigcup_{C\in \cC} T_C$, i.e. the graph obtained after removing the trees corresponding to clusters. We can now recursively call $\textsc{TreeEmbed}(\bar G)$ which embeds $\bar G$ into a tree $T'$, and for all $(u, v) \in E(\bar G)$,
    \[d_{G'}(u, v) \leq d_{T'}(u, v) \quad \text{and} \quad \E[d_{T'}(u, v)] \leq c^{\log_2 (\lfloor k/2 \rfloor + 1)} \cdot d_{G'}(u, v)\]
    which implies that $d_{G}(u, v) \leq d_{T'}(u, v)$ and $\mathbb{E}[d_{T'}(u, v)] \leq c^{\log_2 k}\cdot d_{G}(u, v)$. We note that each edge participates in at most $\log \lfloor k/2\rfloor$ calls of $\textsc{TreeEmbed}$, incurring constant factor distortion each time.

    Finally, we can invoke~\Cref{lem:edge-sep} to glue together the tree $T'$ and $\bigcup_{C\in\cC} T_C$ to obtain tree $T$, such that for all $u, v\in V(G)$
    \[d_G(u, v) \leq d_T(u, v) \quad \text{and} \quad \E[d_T(u, v)] \leq c^{\log_2 k} \cdot d_G(u, v),\] as desired.    
\end{proof}

\subsection{Proving \Cref{lem:embed-outerplanar}}

We now state two important lemmas that will let us prove \Cref{lem:embed-outerplanar}.

\begin{restatable}{lemma}{steinerGraph}
\label{lem:g-with-steiner}
Let $G$ be a $(k/2)$-outerplanar graph, and let $T$ denote the outer face of $G$.
Stochastic procedure $\textsc{LowDiameterDecomposition}(G, T)$ 
returns a plane graph $G_S$ with $V(G) \subseteq V(G_S)$ such that:
    \begin{enumerate}
        \item $G_S$ comprises the face $T$, some Steiner vertices (connected by edges to $T$ and to each other), and some $(k/2)$-outerplanar subgraphs called \EMPH{clusters} each attached to a Steiner vertex $v_C$ by a single~edge.
        \item For every pair of vertices $(u,v)$ in $G$, we have $d_G(u,v) \le d_{G_S}(u,v)$.
        \item For every pair of vertices $(u,v)$ in $G$, we have $\mathbb{E}[d_{G_S}(u,v)] \le O(d_G(u,v))$.
        \item
        Let $(a,b)$ be an edge in $G_S$ such that $a$ and $b$ are either Steiner vertices or vertices in $T$. Let $t_a \in T$ (resp.\ $t_b \in T$) be vertex in $T$ that is closest to $a$ (resp.\ $b$) in $G_S$. Then 
        \[w(a,b) \ge \frac{d_{G_S}(a,t_a) + d_{G_S}(t_a,t_b) + d_{G_S}(t_b, b)}{O(1)}\]
        where $w(\cdot)$ is the weight function for $G_S$.
    \end{enumerate}
\end{restatable}

\begin{restatable}{lemma}{SPR}
\label{lem:SPRB}
    There is an absolute constant $c>1$ such that the following holds:
    Let $\hat{G}$ be the graph formed after the Steiner point removal process.
    For any two vertices $u, v\in V(\hat{G})$, we have:
    \begin{align*}
        d_{\hat{G}}(u, v) \geq d_{G_S}(u, v) 
        \quad \text{and} \quad
        \mathbb{E}[d_{\hat{G}}(u, v)] \leq c \cdot d_{G_S}(u, v)
    \end{align*}
\end{restatable}

\begin{proof}[of \Cref{lem:embed-outerplanar}]
$\textsc{EmbedKOuterplanar}(H, T)$ first runs \textsc{LowDiameterDecomposition}$(H,T)$ to obtain the Steiner graph $H_S$ and clusters $\cC$, 
then perform \textsc{SteinerPointRemoval}$(H_S,\cC, T)$ to produce the output graph $\hat H$.
By \Cref{lem:g-with-steiner}, the plane graph $H_S$ comprises the same outer face $T$ and some $(k/2)$-outerplanar clusters, each attached to a Steiner vertex $v_C$ by a single edge.
Because $V(H_S) = V(H) \cup \Set{v_C : C \in \cC}$ and the Steiner point removal procedure removes all vertices in $\Set{v_C : C \in \cC}$, $V(\hat{H}) = V(H)$.

Consider the graph $\hat{H} \setminus \bigcup_{C \in \cC} C$.  Because \textsc{SteinerPointRemoval} only removes $\Set{v_C : C \in \cC}$ from $H_S$, this graph should be equivalent  to performing SPR on $H_S \setminus \bigcup_{C \in \cC} C$, which must be an outerplanar graph on $V(T)$, which comprises the outer face $T$ and some chords between vertices of $T$.
If we put the clusters in $\cC$ back, now in $\hat{H}$ each $(k/2)$-outerplanar cluster in $\cC$ is attached to some vertex of $T$ by a single edge.

As for the distance between an arbitrary pair of vertices $(u,v)$ in $\hat{H}$, by \Cref{lem:g-with-steiner} we have 
\[
d_{H}(u,v) \le \mathbb{E}[d_{H_S}(u,v)] \le O(1) \cdot d_H(u,v).
\]
By \Cref{lem:SPRB}, it follows that
\[
d_{H}(u,v) \le \mathbb{E}[d_{\hat{H}}(u,v)] \le O(1) \cdot d_{H}(u,v).
\]
\aftermath
\end{proof}
It remains to prove the correctness of \Cref{lem:g-with-steiner} and \Cref{lem:SPRB}.

\section{Creating a Steiner Graph by Clustering}\label{sec:ldd}
In this section, we prove \Cref{lem:g-with-steiner}.

\steinerGraph*
We provide some intuition for the proof of \Cref{lem:g-with-steiner}.
We partition the graph into \emph{clusters}, in a way that cuts any edge with low probability;  from this clustering, we construct the \emph{Steiner graph} by replacing each cluster $C$ with a Steiner vertex $v_C$ and attaching $C$ to $v_C$ with a single edge.
The first 3 properties of \Cref{lem:g-with-steiner} follow by clustering the graph with an appropriate \emph{low-diameter decomposition} at various scales. The 4th property is more difficult to prove. We will end up setting the weight of an edge $(a,b)$ in $G_S$ as $w(a,b) \ge \max \set{d_G(a,T), d_G(b,T)}$, so it suffices to show that $O(1) \cdot d_G(a,T) \ge d_{G_S}(a,T)$ and $O(1) \cdot d_G(b,T) \ge d_{G_S}(b,T)$. While Property (3) says that these inequalities hold \emph{in expectation}, it is crucial that we obtain a deterministic bound. As described in the introduction, we do this by exploiting the \emph{stochastic shortcut partition} of \cite{CCCLPP25}.

The rest of this section is devoted to the proof of \Cref{lem:g-with-steiner}. 
Let $G = (V, E, w_G)$ with outer face $T$. The vertices in $T$ are called \EMPH{terminals}.

\subsection{Step 1: Creating rings} 

Sample a value \EMPH{$\alpha$} uniformly at random from the range $(0,1)$.
We define a set $\EMPH{\cR} \coloneqq (R_0, R_1, \ldots)$ of \EMPH{rings} as follows: for every $i \in \mathbb{N}_{0}$, define $i$-th ring 
\[
\textcolor{BrickRed}{\emph{$R_i$}} \coloneqq \set{v \in V(G) : 2^{i- \alpha} \le d_G(v,T) < 2^{i + 1 - \alpha}}.
\] 
Without loss of generality every edge in $G$ has weight at least 1, the set $\cR$ is a partition of $V(G) \setminus V(T)$.
\begin{lemma}
\label{lem:rings}
    For every edge $(u,v)$ in $G$, the endpoints belong to different rings with probability at most 
    \[ 
    O(1) \cdot \frac{w_G(u,v)}{\min \set{d_G(u,T), d_G(v,T)}}.
    \]
\end{lemma}
\begin{proof}
Let $\gamma_u = d_G(u,T)$ and $\gamma_v = d_G(v,T)$, and assume without loss of generality that $\gamma_u \le \gamma_v$. Triangle inequality implies that $d_G(u,v) \ge \gamma_v - \gamma_u$. We aim to show that $(u,v)$ is cut with probability at most
\(O(1) \cdot \frac{\gamma_v - \gamma_u}{\gamma_u}.\)
Indeed, edge $(u,v)$ is cut iff there exists $i\in\N$ such that $2^{i - \alpha} \in [\gamma_u, \gamma_v)$.
Rearranging, we find that the bad event occurs if there exists $i \in \N$ such that $i - \alpha \in [\log \gamma_u, \log \gamma_v) $.
As $\alpha$ is uniform over $(0,1)$, the bad event occurs
with probability equal to the length of the interval $[\log \gamma_v, \log \gamma_u)$; that is, the probability is
\[\log \gamma_v - \log \gamma_u = \log \frac{\gamma_v}{\gamma_u} = \log \left(1 + \frac{\gamma_v - \gamma_u}{\gamma_u}\right) \le 2 \cdot \frac{\gamma_v - \gamma_u}{\gamma_u}.\]
\aftermath
 \end{proof}

\subsection{Step 2: Chopping rings into clusters}

We now partition each ring into a set of clusters. The final guarantee on the clusters is given by \Cref{lem:clustered-rings}. We remark that last guarantee of \Cref{lem:clustered-rings} is somewhat awkward, but is essential to proving \Cref{lem:deterministic-gamma} later: in the Steiner graph, the distance between a vertex $v$ and $T$ is preserved \emph{deterministically} with only $O(1)$ distortion.

Given a partition of $G$ into clusters $\cC$, define the \EMPH{cluster graph $G/\cC$} to be the graph obtained by contracting each cluster $C_i \in \cC$ into a \EMPH{Steiner vertex $s_i$}. 
Sometimes we abuse notation and don't distinguish between the cluster $C_i$ and the Steiner vertex $s_i$ representing the cluster. We say that a cluster $C_i$ (or its Steiner vertex $s_i$) has \EMPH{nontrivial intersection} with a path $\pi$ in $G$ if $C_i \cap \pi \neq \varnothing$.

\begin{lemma}
\label{lem:clustered-rings}
    There is a stochastic procedure that partitions $V(G)$ into \EMPH{clusters $\cC$}
    such that:
    \begin{enumerate}
        \item Every terminal $t \in T$ is a singleton cluster $\set{t} \in \cC$; collectively these clusters are called \EMPH{terminal clusters}. 
        \item For \EMPH{non-terminal cluster} $C \in \cC$, there is some $i$ such that $C$ is contained in ring $R_i \in \cR$ (meaning that $d_G(v,T) \in [2^{i-\alpha}, 2^{i+1-\alpha})$ for every $v \in C$) and $\diam(C) \le 2^i$.
        \item Every edge $(u,v)$ is a \EMPH{cut edge}\footnote{that is, $u$ and $v$ belong to different clusters} with probability at most $O(1) \cdot \frac{d_G(u,v)}{\min \set{d_G(u,T), d_G(v,T)}}$.
        \item
        For any vertex $v$ in $V(G) \setminus V(T)$, 
        there exists a path $P$ in the cluster graph $G / \cC$ between the cluster containing $v$ and some terminal cluster with the following properties
        \begin{itemize}
            \item
            Let $i$ denote the index of the ring $R_i \in \cR$ with $v \in R_i$. Path $P$ comprises a single terminal cluster as well as clusters in rings $R_j$ with $j \le i$.
            \item For every $j \in [i]$, there are only $O(1)$ clusters in the ring $R_j$ that are in $P$.
            \item Let $(s_1, s_2)$ be an edge in  $P$, let $C_1$ and $C_2$ be the clusters associated with $s_1$ and $s_2$, and let $k$ be the maximum index such that either $C_1$ or $C_2$ is in the ring $R_k$. Then there exists $a \in C_1$ and $b \in C_2$ with $d_G(a,b) \le O(2^k)$.
        \end{itemize}
        Path $P$ is called a \EMPH{$v$-to-$T$ preserving path}.
    \end{enumerate}
\end{lemma}

The key lemma for the construction of \Cref{lem:clustered-rings} is the \emph{stochastic} shortcut partition of \cite{CCCLPP25}.

\begin{lemma}[Lemma 6.4 of \cite{CCCLPP25}]
\label{lem:ccclpp}
    Let $H$ be a planar (or minor-free) graph, and let $\Delta$ be a parameter. There is a stochastic procedure that partitions $V(H)$ into a set of clusters $\cC$, such that 
    \begin{itemize}
        \item every cluster in $\cC$ has (strong) diameter at most $\Delta$,
        \item every edge $(u,v)$ in $H$ is cut with probability at most $O(1) \cdot d_H(u,v)/\Delta$, 
        \item for every pair of vertices $u,v$ in $H$ and any shortest path $\pi$ in $H$ between $u$ and $v$, there exists a path $P$ in the unweighted cluster graph $H \setminus \cC$ between the cluster containing $u$ and the cluster containing $v$, such that (1)~$P$ has length at most $O(d_H(u,v) / \Delta)$, and (2) every cluster in $P$ has nontrivial intersection with $\pi$.
    \end{itemize}

\end{lemma}

\begin{proof}[of \Cref{lem:clustered-rings}]

For every terminal $t \in T$, create a singleton cluster $C = \set{t}$.
For every ring $R_i \in \cR$,
apply \Cref{lem:ccclpp} on $G[R_i]$ to construct a partition of $R_i$ into clusters. Overall, the set $\cC$ of all clusters is a partition of $V(G)$.
Properties 1, 2, and 3 of \Cref{lem:clustered-rings} are immediate from the construction of $\cC$. It remains to prove property 4. 
We prove, by induction on the ring index $i$, that for every vertex $v$ in $R_i$, there exists a $v$-to-$T$ preserving path $P$.  

Let $\pi$ be a shortest path in $G$ from $v$ to some terminal $t \in T$. Every vertex on $\pi$ is no farther from $T$ than $v$; thus, every vertex on $\pi$ (other than $t$) is in some ring $R_j$ with $j \le i$.
Walking along $\pi$ from $u$ to $T$, let \EMPH{$x$} be the first vertex on $\pi$ that is not in $R_i$, and let \EMPH{$x'$} be the vertex immediately before $x$ in $\pi$.

We define a path \EMPH{$P_{<i}$} in the cluster graph as follows. If $x$ is a terminal vertex, ie $x \in T$, then define $P_{<i}$ to be the trivial path that contains only the singleton cluster $\set{x}$. Otherwise, by induction, define $P_{<i}$ to be some $x$-to-$T$-preserving path (which runs between the cluster containing $x$ and some terminal cluster).
We now define a path \EMPH{$P_i$} between the cluster containing $v$ and the cluster containing $x'$. Let $\pi[v:x']$ denote the subpath of $\pi$ between $v$ and $x'$; observe that it is a shortest path between $v$ and $x'$ in $G[R_i]$, with length at most $\norm{\pi} \le 2^{i+1 - \alpha} = O(2^i)$.
Let $G[R_i] / \cC$ denote the (unweighted) cluster graph obtained from $G[R_i]$ by contracting every cluster in $\cC$ that is in $R_i$.
By \Cref{lem:ccclpp}, there is some path $P_i$ in the $G[R_i] / \cC$ between the cluster containing $v$ and the cluster containing $x'$, such that $P_i$ has length at most 
$O(\norm{\pi[v:x']}/2^i) = O(1)$, and every cluster in $ P_i$ has nontrivial intersection with $\pi[v:x']$. 
We claim that the concatenation
\[\textcolor{BrickRed}{\emph{$P$}} \coloneqq  P_i \circ  P_{<i} \]
is a $v$-to-$T$ preserving path.
First observe that $P$ is indeed a path in the cluster graph $G / \cC$: this is because the cluster containing $x'$ is adjacent in $G / \cC$ to the cluster containing $x$ (and thus $P$ is a walk), and $P_i$ is disjoint from $P_{<i}$ (and thus $P$ is a path).
Now, clearly $P$ contains a single terminal cluster and clusters in rings $R_j$ with $j \le i$. Because $P_i$ has length $O(1)$, there are only $O(1)$ clusters in ring $R_i$ that are in $P$; by the inductive construction of $P_{<i}$, for every $j < i$ there are only $O(1)$ clusters in ring $R_j$ that are in $P$.
Finally, let $(s_1, s_2)$ be an edge in  $P$, let $C_1$ and $C_2$ be the associated clusters, and let $k$ be the maximum index such that either $C_1$ or $C_2$ is in the ring $R_k$; we need to show that there exists $a \in C_1$ and $b \in C_2$ with $d_{G}(a,b) \le 2^k$. 
Indeed, if $(s_1, s_2)$ is an edge in $P_{<i}$, then the claim is true by induction. Otherwise, $(s_1,s_2)$ is an edge that either connects two clusters in $P_i$ or connects the cluster containing $x'$ with the cluster containing $x$;
in either case, $\pi$ intersects nontrivially with both $C_1$ and $C_2$, so we can choose $a \in C_1 \cap \pi$ and $b \in C_2 \cap \pi$, and we have $d_{G}(a,b) \le \norm{\pi} \le O(2^i) = O(2^k)$ as desired.
\end{proof}

\subsection{Step 3: Construction of the Steiner graph}\label{sec:steiner-graph-construction}

We now construct the graph $G_S$, whose vertex set consists of the original vertices in $V(G)$ and Steiner vertices associated with the clusters in $\mathcal{C}$. 
We define the \EMPH{Steiner graph $G_S$} as follows:
\begin{enumerate}
    \item Initialize $G_S$ as the 
    cluster graph $G_S / \cC$, with vertex set $\set{s_i : C_i \in \cC}$.
    (We only create Steiner vertices for the non-terminal clusters, and each terminal cluster $C_i = \set{t_i}$ is represented by the terminal $t_i$ itself.)
    \item For every \emph{non-terminal} cluster $C_i \in \cC$, choose an arbitrary vertex \EMPH{$r_i$} on the outermost layer of $G[C_i]$. Add the subgraph $G[C_i]$ to $G_S$, connected by a 0-weight edge $(s_i, r_i)$.
    \item 
    For every \emph{terminal} cluster $C_i = \set{t_i}$,
    for the sake of notation, we define both $s_i$ and $r_i$ to be $t_i$.
    \item For every pair of clusters $C_i, C_j$ in $\cC$ with an edge $(s_i, s_j)$ in $G_S$, 
    set the weight of the edge to be 
    \[
    w(s_i, s_j) \coloneqq 
    \max_{\substack{x\in C_i\\y\in C_j}}d_G(x,y) + \max_{\substack{v\in C_i \cup C_j}} d_G(v, T)
    \] 
\end{enumerate}
This completes the description of $G_S$. It remains to prove the four properties of \Cref{lem:g-with-steiner}.

\begin{observation}[Property 1 of \Cref{lem:g-with-steiner}]
    Steiner graph $G_S$ is a plane graph with outer face $T$. It comprises the face $T$, some Steiner vertices (connected to edges by $T$ and to each other), and some $(k/2)$-outerplanar subgraphs each attached to a Steiner vertex by a single edge.
\end{observation}

\begin{lemma}[Property 2 of \Cref{lem:g-with-steiner}]
For every pair of vertices $(u,v)$ in $G$, we have $d_G(u,v) \le d_{G_S}(u,v)$.
\end{lemma}

\begin{proof}
Let $P$ be the shortest path between $u$ and $v$ in $G_S$. We show $\norm{P} \ge d_G(u,v)$. There are two cases.

\begin{itemize}
\item
\textit{If $u$ and $v$ are in the same cluster $C \in \cC$}: 
Then $P$ lies in $G[C]$. (Indeed, if $P$ wandered outside of $G[C]$, then $P$ would have to leave and reenter $C_i$ at the vertex $s_i$, forming a cycle.) Thus $\norm{P} = d_{G[C]}(u,v) \ge d_G(u,v)$.

\item
\textit{Now suppose that $u$ and $v$ are in different clusters, say $C_{u}$ and $C_{v}$}:  
Let $s_u$ and $s_v$ be the Steiner vertices associated with $C_u$ and $C_v$, respectively. In this case, the prefix of $P$ is a path from $u$ to $s_{u}$, the suffix of $P$ is a path from $s_{v}$ to $v$, and every other vertex on $P$ is a Steiner vertex. (Otherwise, $P$ would contain a cycle.)
Let $\hat P = [s_u = s'_0, s'_1, \ldots, s'_\ell = s_v]$ be the subpath of $P$ that runs between $s_u$ and $s_v$.
For every $i \in [\ell]$, let \EMPH{$C'_i$} be the cluster associated with $s'_i$, and let $\textcolor{BrickRed}{\emph{$z_i$}}\in C'_i$ be an arbitrary vertex in $C'_i$ (choosing $z_0 = u$ and $x_\ell = v$). Then we have
\begin{align*}
    \norm{P} \ge \norm{\hat P}
    &= \sum_{i = 0}^{\ell - 1} w(s_i, s_{i+1})\\
    &\ge \sum_{i = 0}^{\ell - 1} \max_{x \in C'_i, y \in C'_{i+1}} d_G(x,y)\\
    &\ge \sum_{i = 0}^{\ell - 1} d_G(z_i,z_{i+1})\\
    &\ge d_G(z_0, z_\ell) = d_G(u,v).
\end{align*}
\end{itemize}
\aftermath
\end{proof}

\begin{lemma}[Property 3 of \Cref{lem:g-with-steiner}]
    For every pair of vertices $(u,v)$ in $G$, we have $\mathbb{E}[d_{G_S}(u,v)] \le O(1) \cdot d_G(u,v)$.
\end{lemma}
\begin{proof}
    It suffices to prove the lemma for every edge $(u,v)$ in $G$.
    Let $\gamma = \dist_G(u,T)$ and $\gamma' = \dist_G(v,T)$. Assume without loss of generality that $\gamma \ge \gamma'$. 
In this proof, we consider three cases. In the first two cases (Case 1 and Case 2a), we prove that $\dist_{G_s}(u,v)= O(d_G(u,v))$. 
In the last case (Case 2b), we prove that $\dist_{G_s}(u,v) = O(\gamma)$, but this case occurs only with probability at most $O(d_G(u,v) / \gamma)$. 
Overall, we conclude that the expected distance $\mathbb{E}[\dist_{G_s}(u,v)]$ is $O(d_G(u,v))$.

\medskip \noindent \textbf{Case 1: \boldmath{$(u,v)$} is not cut.} Suppose that both $u$ and $v$ belong to the same cluster $C \in \cC$. The edge $(u,v)$ is in the induced subgraph $G[C]$ and thus is also in the graph $G_S$, and we conclude that the path $P(u,v) = [u,v]$ is in $G_S$ and satisfies $\norm{P(u,v)} = d_G(u,v)$.

\medskip \noindent \textbf{Case 2: \boldmath{$(u,v)$} is cut.}
Suppose otherwise: the edge $(u,v)$ is cut.
We claim that 
\begin{equation}
\label{eq:detour-len}
    \norm{d_{G_S}(u,v)} \le d_G(u,v) + O(\gamma).
\end{equation}
Let $C_u$ and $C_v$ be the clusters containing $u$ and $v$, respectively, with Steiner vertices $s_u$ and $s_v$. By \Cref{lem:clustered-rings}(2), we have $\diam(G[C_u]) \le O(\gamma)$ $\diam(G[C_v]) \le O(\gamma') \le O(\gamma)$, so the construction of $G_S$ implies that $d_{G_S}(u,s_u) \le O(\gamma)$ and $d_{G_S}(s_v,v) \le O(\gamma)$.
By construction, the edge $(s_u, s_v)$ in $G_S$ has weight equal to $\max_{x \in C_u, y\in C_v} d_G(x,y) + O(\gamma)$.
But triangle inequality, together with the fact that $C_u$ and $C_v$ have diameter at most $O(\gamma)$, implies that for any $x \in C_u$ and $y \in C_v$, we have $d_G(x,y) \le d_G(x, u) + d_G(u,v) + d_G(v, y) \le d_G(u,v) + O(\gamma)$.  Thus $d_{G_S}(s_u, s_v) = d_G(u,v) + O(\gamma)$.
This completes the proof of \Cref{eq:detour-len}.
From here, there are two cases.

\medskip \noindent \textbf{Case 2a: \boldmath{$d_G(u,v) \ge \gamma/2$}.} In this case, \Cref{eq:detour-len}  implies $P(u,v)$ has length at most $O(d_G(u,v))$, and we are done. 

\medskip \noindent \textbf{Case 2b: \boldmath{$d_G(u,v) < \gamma/2$}.} In this case, \Cref{eq:detour-len} implies that $\norm{P} = O(\gamma)$. We now show that Case 2b occurs with probability at most $O(d_G(u,v) / \gamma)$. Recall that $\gamma = d_G(u,T)$, $\gamma' = d_G(v,T)$, and $\gamma > \gamma'$. We claim that
\begin{equation}
\label{eq:gamma-similar}
\gamma < 2 \gamma'.
\end{equation}
Indeed, triangle inequality 
implies that $\gamma = \dist_G(u,T) \le \dist_G(u,v) + \dist_G(v,T)$; after substituting $\dist_G(u,v) < \gamma/2$ and $\dist_G(v,T) = \gamma'$, we conclude $\gamma / 2 < \gamma'$ as desired. 
Now, by \Cref{lem:clustered-rings}(3), we have that $(u,v)$ is cut with probability at most $O(d_G(u,v) / \min\set{\gamma,\gamma'})$. By \Cref{eq:gamma-similar}, this probability is at most $O(d_G(u,v) / \gamma)$, as desired.
\end{proof}

It remains to prove \Cref{lem:g-with-steiner}(4). This property follows as a simple corollary of another lemma: there is a \emph{deterministic} bound on the distance between a vertex $v$ and its closest terminal in $G_S$.

\begin{lemma}
\label{lem:deterministic-gamma}
    Let $C \in \cC$ be a cluster with associated Steiner vertex $s \in V(G_S)$. Then
    \[d_{G_S}(s,T) \le O(1) \cdot \max_{v \in C} d_G(v,T)\]
\end{lemma}
\begin{proof}
    Let $\textcolor{BrickRed}{\emph{$\gamma$}} \coloneqq d_G(v,T)$. Let $C \in \cC$ be the cluster containing $v$, let $r$ be the root vertex of $C$ and let $s$ be the Steiner vertex of $C$.
    By \Cref{lem:clustered-rings}(2), the diameter of $G[C]$ is $O(\gamma)$, and so $d_{G_S}(v, r) = d_{G[C]}(v, r) \le O(\gamma)$. There is a 0-weight edge between $r$ and $s$ in $G_S$. To prove the lemma, it suffices to shows that $d_{G_S}(s, T) \le O(\gamma)$.
    By \Cref{lem:clustered-rings}, there is a $v$-to-$T$ preserving path \EMPH{$P$} in the cluster graph $G / \cC$ between $s$ and $T$. We view path $P$ as a path in $G_S$. The rest of this proof shows that $\norm{P} = O(\gamma)$.

    For any edge $e = (s_1, s_2)$ in $P$, letting $C_1$ and $C_2$ denote the clusters associated with $s_1$ and $s_2$, we define the \EMPH{scale} of $e$ to be the maximum index $i$ such that $C_1$ or $C_2$ is in the ring $R_i$.%
    \footnote{In this definition, we abuse notation and treat terminal clusters as belonging to the ring with index $-1$. Observe that every edge $e$ in $P$ has at least one non-terminal endpoint, and so the scale of $e$ is always $\ge 0$.},
    Let \EMPH{$\mathcal{E}_i$} denote the set of scale-$i$ edges in $P$. Let \EMPH{$\kappa$} denote the index such that $v$ is in ring $R_\kappa$. Observe that $2^\kappa = O(\gamma)$ by \Cref{lem:clustered-rings}(2). Moreover, every edge in $P$ has scale $j$ with $j \le \kappa$, by definition of $v$-to-$T$ preserving path.
    The sets $\mathcal{E}_i$ for $i \in [\kappa]$ form a partition of the edges of $P$:
    \begin{equation*}
        \norm{P} = \sum_{\text{edge } e \in P} w(e)
        = \sum_{i \in [\kappa]} \sum_{e \in \mathcal{E}_i} w(e)
    \end{equation*}
    By definition of $v$-to-$T$ preserving path, there are $O(1)$ edges in $\mathcal{E}_i$, for all $i \in [\kappa]$.
    We show below that every edge in $\mathcal{E}_i$ has weight at most $O(2^i)$. This is sufficient to prove the lemma:
    \[
    \norm{P} = \sum_{i \in [\kappa]} O(2^i) = O(2^\kappa) = O(\gamma).
    \]

    It remains to show that an arbitrary edge $(s_1, s_2)$ in $\mathcal{E}_i$ has weight $O(2^i)$, for any $i \in [\kappa]$. Indeed, let $C_1$ and $C_2$ denote the clusters associated with $s_1$ and $s_2$. The weight of edge $(s_1, s_2)$ in $G_S$ is 
    \[w(s_1, s_2) = \max_{x \in C_1, y \in C_2} d_G(x,y) + \max_{v \in C_1 \cup C_2} d_G(v, T) \le \max_{x \in C_1, y \in C_2} d_G(x,y) + O(2^i).\]
    By definition of $v$-to-$T$ preserving path, there exist $a \in C_1$ and $b \in C_2$ with $d_G(a,b) \le O(2^i)$. Moreover, $\diam(G[C_1]) \le O(2^i)$ and $\diam(G[C_2]) \le O(2^i)$, by \Cref{lem:clustered-rings}(2). Thus, triangle inequality implies that for any $x \in C_1$ and $y \in C_2$, we have $d_G(x,y) \le d_G(x, a) + d_G(a, b) + d_G(b, y) \le O(2^i)$. Thus $w(s_1,s_2) \le O(2^i)$ as desired.
\end{proof}

\begin{lemma}[Property 4 of \Cref{lem:g-with-steiner}]
        Let $(a,b)$ be an edge in $G_S$ such that $a$ and $b$ are either Steiner vertices or vertices in $T$. Let $t_a \in T$ (resp.\ $t_b \in T$) be vertex in $T$ that is closest to $a$ (resp.\ $b$) in $G_S$. Then 
        \[w(a,b) \ge \frac{d_{G_S}(a,t_a) + d_{G_S}(t_a,t_b) + d_{G_S}(t_b, b)}{O(1)}.\]
\end{lemma}
\begin{proof}
To simplify notation, define $\gamma \coloneqq \max_{v \in C_a \cup C_b} d_G(v,T)$. 
By construction, we have $w(a,b) = \max_{x \in C_a, y \in C_b} d_G(x, y) + \gamma$.
By \Cref{lem:deterministic-gamma}, we have $d_{G_S}(a, t_a) \le O(\gamma) \le O(w(a,b))$ and $d_{G_S}(t_b, b) \le O(\gamma) \le O(w(a,b))$. 
Moreover, triangle inequality implies that $d_{G_S}(t_a, t_b) \le d_{G_S}(t_a, a) + w(a,b) + d_{G_S}(b, t_b) \le O(w(a,b))$. It follows that $d_{G_S}(a, t_a) + d_{G_S}(t_a,t_b) + d_{G_S}(t_b, b) \le O(w(a,b))$ as desired.
\end{proof}

\section{Steiner Point Removal}

We are now in the final phase, where we construct graph $\hat{G}$ in the following manner: Consider the graph $G_S$, i.e. the graph furnished by the output of the procedure in~\Cref{sec:ldd} with the guarantees of~\Cref{lem:g-with-steiner}. Temporarily remove all the clusters from $G_S$ to get $G_S\setminus \cC$, and then apply the Steiner point removal technique~\cite{chang2024shortcut}, \textsc{PlanarSPR}, to $G_S\setminus\cC$. Finally attach back the clusters in $\cC$ to the terminal $t'$ to which the corresponding Steiner point is associated via root $r\in V(C)$, with direct edge of weight $w_{\hat{G}}(r, t') = d_{G_S}(r, t')$. 
We restate~\Cref{lem:SPRB} here.
\SPR*

In order to prove~\Cref{lem:SPRB}, we first state the guarantees of the Steiner point removal technique~\cite{chang2024shortcut}, which we use a black box. 

\begin{lemma}[Rephrasing of Theorem 1.4 of~\cite{chang2024shortcut}]
\label{lem:SPRTh1.2}
    Let $G=(V,E,w)$ be an arbitrary edge-weighted planar graph and let $T\subseteq V$ be a set of terminals. Then there is solution of the SPR problem with distortion $O(1)$.
    In other words, there is a graph minor $M$ of $G$ such that $V(M) = T$ and for some constant $\alpha\geq 1$, we have that for every pair of terminals $t_1, t_2\in T$,
    \[d_G(t_1, t_2)\leq d_M(t_1, t_2)\leq \alpha\cdot d_G(t_1, t_2)\]
\end{lemma}
We also require the following guarantee from Claim 2.7 of~\cite{chang2024shortcut}:

\begin{lemma}[Claim 2.7 of~\cite{chang2024shortcut}] 
\label{lem:SPRCl2.7}
    Let $v\in V(G)$ be such that $v$ gets associated with the terminal $t\in T$ after the Steiner point removal process. We have
    \[d_G(v, t)\leq O(1)\cdot d_G(v, T)\]
\end{lemma}

Let $\hat{G}$ be the graph formed after the Steiner point removal process. 
Note that $\hat{G}$ is an outerplanar layer on the terminal set $T$ with a bunch of $(k/2)$-outerplanar clusters hanging off terminals in $T$ via edge connections. Let $f: V(G_S) \to V(G)$ be the function which maps a Steiner vertex in $G_S$ to the terminal to which it is attached in $\hat G$ (and maps a non-Steiner vertex to itself), according to the black box procedure $\textsc{PlanarSPR}$. Now we are ready to prove the final bound.

\subsection{Distortion guarantee for SPR}

We first prove the following distortion bound for any edge between Steiner vertices or terminals in $G_S$.
\begin{claim}\label{lem:SPR-local}
    Let $a,b\in V(G_S)$ be Steiner vertices such that $(a,b)\in E(G_S)$. Moreover, let $f(a)\in T$ (resp.\ $f(b)\in T$) be the terminal vertex that the vertex $a$ (resp.\ $b$) associates with after the Steiner point removal. 
    Then 
    \[d_{G_S}(f(a), f(b)) \leq O(1)\cdot w_{G_S}(a, b).\]
\end{claim}
\begin{proof}
    Let $t_a \in T$ (resp.\ $(t_b\in T)$ be the terminal vertex closest to $a$ (resp.\ $b$) in $G_S$. Consider the walk $W_{G_S}(f(a), f(b))  = W_{G_S}(f(a)\rightsquigarrow a \rightsquigarrow t_a \rightsquigarrow t_b \rightsquigarrow b \rightsquigarrow f(b))$ in $G_S$, obtained by concatenating the shortest paths between successive vertices $(f(a), a, t_a, t_b, b, f(b))$. We apply~\Cref{lem:SPRCl2.7} to get:
    \begin{align*}
        d_{G_S}(f(a), f(b)) &\leq \len(W_{G_S}(f(a), f(b))) \\
        &= d_{G_S}(f(a), a) + d_{G_S}(a, t_a) + d_{G_S}(t_a, t_b) + d_{G_S}(t_b, b) + d_{G_S}(b, f(b))\\
        &\leq O(1)\cdot d_{G_S}(a, t_a) + d_{G_S}(t_a, t_b) + O(1)\cdot d_{G_S}(t_b, b)\\
        &\leq O(1)\cdot (d_{G_S}(a, t_a) + d_{G_S}(t_a, t_b) + d_{G_S}(t_b, b))\\
        &\leq O(1)\cdot w_{G_S}(a, b)\quad\quad \text{by Property 4 of~\Cref{lem:g-with-steiner},}
    \end{align*}
    as claimed.
\end{proof}
\begin{proof}[of \Cref{lem:SPRB}] 
    \textbf{Constant factor distortion:} We prove that $d_{\hat{G}}(u, v) \leq O(1)\cdot d_{G_S}(u, v)$ for all $u,v \in V(G)$. Let $s_u$ (resp.\ $s_v$) be the Steiner vertex associated with the cluster $C_u$ (resp.\ $C_v$) containing $u$ (resp.\ $v$). 
    
    We now look at the shortest path $\pi_{G_S}(u, v) = (u\rightsquigarrow s_u - x_1-x_2-\ldots-x_L-s_v\rightsquigarrow v)$, where each $x_i$ is either a Steiner vertex or a vertex in $G_\text{out} \cup T$, and $(x_i, x_{i+1})\in E(G_S)$ for all $0\leq i\leq L$. Consider the corresponding walk in $\hat G$: $W_{\hat{G}}(u, v) = W_{\hat{G}}(u\rightsquigarrow f(s_u) \rightsquigarrow f(x_1) \rightsquigarrow\ldots \rightsquigarrow f(x_L) \rightsquigarrow f(s_v)\rightsquigarrow v)$, obtained by concatenating shortest paths $\pi_{\hat G}(u, f(s_u))\cup\{\pi_{\hat G}(x_i, x_{i+1})\}_{i=0}^L\cup \pi_{\hat G}(f(s_v), v)$, where we set $s_u := x_0$ and $s_v:=x_{L+1}$ for convenience. We have:
    \begin{align}
        d_{\hat{G}}(u, v) \leq \len(W_{\hat{G}}(u, v)) = \len(\pi_{\hat{G}}(u, f(x_0))) + \sum_{i=0}^L \len(\pi_{\hat{G}}(f(x_i), f(x_{i+1}))) + \len(\pi_{\hat{G}}(x_{L+1}, v))\label{eq:spr-ub-main-eqn}
    \end{align}
    Applying the weight setting $w_{\hat G}(r_a, f(s_a)) = d_{G_S}(s_a, f(s(a))$, and using~\Cref{lem:SPRCl2.7} and Property 4 of~\Cref{lem:g-with-steiner},
    \begin{align*}
        \len(\pi_{\hat G}(u, f(x_0)) &= d_{\hat G}(u, f(s_u)) = d_{G_S}(u, r_u) + w_{\hat G}(r_u, f(s_u)) \\
        &= d_{G_S}(u, r_u) + d_{G_S}(s_u, f(s_u))\\
        &\leq d_{G_S}(u, r_u) + O(1)\cdot d_{G_S}(s_u, t_{s_u})\\
        &\leq d_{G_S}(u, r_u) + O(1)\cdot w_{G_S}(x_0, x_1)
    \end{align*}
    Similarly, we can also bound \[\len(\pi_{\hat G}(f(x_{L+1}), v) \leq d_{G_S}(v, r_v) + O(1)\cdot w_{G_S}(x_L, x_{L+1})\] Moreover,~\Cref{lem:SPR-local} lets us bound \[\len(\pi_{\hat{G}}(f(x_i), f(x_{i+1}))) \leq O(1)\cdot w_{G_S}(x_i, x_{i+1})\] for all $0\leq i\leq L$. Putting it all together in Equation~\eqref{eq:spr-ub-main-eqn}, we have:
    \[d_{\hat G}(u, v) \leq d_{G_S}(u, r_u) + O(1)\cdot \sum_{i=0}^L w_{G_S}(x_i, x_{i+1}) + d_{G_S}(r_v, v) \leq O(1)\cdot d_{G_S}(u, v)\]
    This completes the proof of bounded distortion.

    \textbf{Non-contraction:} We now prove the non-contraction. Let $\bar\pi_{\hat{G}}(u, v)$ be the shortest path between $u$ and $v$ in $\hat{G}$. 
    By~\Cref{lem:SPRTh1.2}, we have
    \begin{align}
        d_{\hat{G}}(u, v) &= \len(\bar\pi_{\hat{G}}(u, v)) = d_{\hat{G}}(u, f(s_u)) + d_{\hat{G}}(f(s_u), f(s_v)) + d_{\hat{G}}(f(s_v), v)\nonumber\\
        &\geq (d_{G_S}(u, r_u) + d_{G_S}(r_u, f(s_u))) + d_{G_S}(f(s_u), f(s_v)) + (d_{G_S}(f(s_v), r_v) + d_{G_S}(r_v, v))\label{eq:spr_lb_1}
    \end{align}
    Now let $\hat\pi_{G_S}(u, v)$ be the shortest path between $u$ and $v$ in $G_S$. By triangle inequality,
    \begin{align}
        d_{G_S}(u, v) &= \len(\hat\pi_{G_S}(u, v)) = d_{G_S}(u, s_u) + d_{G_S}(s_u, s_v) + d_{G_S}(s_v, v) \nonumber\\
        &\leq d_{G_S}(u, r_u) + (d_{G_S}(s_u, f(s_u)) + d_{G_S}(f(s_u), f(s_v)) + d_{G_S}(f(s_v), s_v)) + d_{G_S}(r_v, v)\label{eq:spr_lb_2}  
    \end{align}
    From Equations~\eqref{eq:spr_lb_1} and~\eqref{eq:spr_lb_2}, we have:
    \[d_{\hat G}(u, v) \geq d_{G_S}(u, v)\]
    This completes the proof of~\Cref{lem:SPRB}.
\end{proof}

\section{Extension to $k$-Shell Instances}

We now consider multiflow instances defined by a set of supply-demand terminal pairs, denoted by $\mathfrak{T} = \{(s_1, t_1), \ldots, (s_{\alpha}, t_{\alpha})\}$ on a graph $G$. 
Every edge $e$ of $G$ is associated with some capacity $c(e)$, and each terminal pair $(s_i, t_i)$ is associated with a demand value $d_i$.
The \EMPH{maximum concurrent flow} is the maximum value $\alpha$ such that $\alpha \cdot d_i$ units of flow can be simultaneously routed from every $s_i$ to $t_i$ (the flows need not be integral) while respecting capacity constraints. The \EMPH{sparsest cut} is the minimum, over all cuts, of the ratio of the capacity of the cut to the demand of the cut; formally, the sparsest cut is
$$\min_{S \subset V(G)} \frac{\sum_{\text{edge $e$ between $S$ and $\bar S$}} c(e)}{\sum_{s_i \in S, t_i \not \in S} d_i}.$$
We are interested in the (fractional) flow-cut gap: the ratio between the value of the sparsest cut and the value of maximum concurrent flow.
We consider \EMPH{$k$-shell instances}, where at least one terminal of each pair is guaranteed to lie within the outer $k$ layers of $G$. 
We prove a tree embedding theorem for these instance, which immediately implies that the fractional flow-cut gap is at most $\poly(k)$ \cite{LLR95, AY98}.

\begin{theorem}\label{thm:kshell}
    Let $G$ be a planar graph with $\mathcal{L} = \{L_1,L_2,\ldots, L_\gamma\}$ the layer decomposition\footnote{that is, $L_1$ is the set of vertices on the outermost face, and $L_i$ is the set of vertices on the outermost face of the graph obtained after removing layers $L_1,\ldots, L_{i-1}$} of $G$. Let $\mathfrak{T} = \{(s_1, t_1),\ldots, (s_\alpha, t_\alpha)\}$ be a set of terminal pairs, such that at least one of each $(s_i, t_i)$ pair lies in the outer $k$ layers.
    Then the shortest path metric of $G$ admits an embedding into a distribution over trees $T^*$ such that
    for each $(s, t)\in \mathfrak{T}$,
    \[d_G(s, t) \leq d_{T^*}(s, t) \quad \text{and} \quad \E[d_{T^*}(s, t)] \leq \poly(k)\cdot d_G(s, t).\]
\end{theorem}

\begin{corollary}
    For any $k$-shell multiflow instance, the fractional flow-cut gap is bounded by $\poly(k)$.
\end{corollary}

At a high level, our approach is the following: we isolate the inner core of the graph (i.e.\ layers $L_{k+1}, \dots$), and apply the $\textsc{PlanarSPR}$ procedure, treating the $k^{\text{th}}$ layer as the terminal set $T$. This transformation yields a $k$-outerplanar graph, with some potentially large clusters hanging off the vertices of $T$. Every terminal $(s,t)$ path contains at least one of $s$ or $t$ in the outer $k$ layers. We embed all these clusters into star graphs with edge weights equal to the distance of the vertex from its root. We can safely execute this since it is guaranteed that both vertices of a terminal pair do not lie within the same cluster. This is now a $k$-outerplanar graph with $O(1)$ distortion. We can finally embed this graph into random trees using our $\poly(k)$ distortion algorithm.
The formal algorithm is described below:

\begin{figure}[h!]
\small
\centering
\begin{algorithm}
% \internallinenumbers
\textul{$\textsc{K-ShellEmbedToTree}(G)$}: \+
\\ if $G$ is $k$-outerplanar: \+
\\ return \textsc{TreeEmbed}(G) \-

\\ $\{L_1,\ldots,L_\gamma\} \gets$ layers of $G$ (where $L_1$ is the outer face)
\\ $G_{in}^1, \ldots, G_{in}^\beta \gets$ connected components of $G[\bigcup_{i=k}^\gamma L_i]$\\
$T^i \gets L_k^i$ for each $i\in [\beta]$

\\ for $i \gets 1$ to $\beta$: \+
\\ $(H_{in}^i, \mathcal{C}^i)\gets \textsc{EmbedKOuterplanar}(G_{in}^i, T^i)$
\\ for each $C\in \mathcal{C}^i$: \+
\\ $\textit{rt} \gets$ the terminal vertex associated with $C$.
\\ $H_{in}^i \gets$ \textsc{MakeStar}($H_{in}^i, C, \textit{rt}$) \-\-
\\ $G' \gets$ the $k$-outerplanar graph formed by attaching $\{H_{in}^i\}_{i=1}^\beta$ with $G[\bigcup_{i=1}^{k-1} L_i]$.
\\ return \textsc{TreeEmbed}($G'$)
\end{algorithm}
\label{fig:kshell_alg}
\end{figure}

In the above procedure, let $\hat H$ denote the intermediate graph formed just before the $\textsc{MakeStar}$ subroutine is invoked. The \EMPH{$\textsc{MakeStar}$} procedure converts each cluster $C$ into a star graph. Specifically, for every vertex $c\in C$, it creates a direct edge to the designated root terminal vertex $\textit{rt}$ with weight $d_{\hat H}(rt, c)$.

\begin{proof}[of \Cref{thm:kshell}]
    We prove that the distortion faced by any $(s, t)$ pair just before the final $\textsc{TreeEmbed}$ step is $O(1)$. Let $G_\text{out} = G[\bigcup_{i=1}^k L_i]$.
    
    \medskip \noindent \textbf{\boldmath Case 1: $s\in G_{\text{out}}, t\in G_{\text{in}}$.} Let $P=P_{G'}(s, t)$ be the shortest path between $s$ and $t$ in $G'$. Let $x$ be the first vertex on $L_k$ when traveling from $s$ to $t$ via $P$. Clearly, $d_{G'}(s,x) = d_G(s, x)$ since all edges on this path belong to $G_{\text{out}}$, and, hence, are untouched. We invoke~\Cref{lem:embed-outerplanar} to bound $d_G(x, t) \leq \E[d_{G'}(x, t)] \leq O(1)\cdot d_G(x, t)$. 

    \medskip \noindent \textbf{\boldmath Case 2: Both $s,t\in G_{\text{out}}$.} The argument is similar to Case 1 here. Let $P$ be the shortest $s$ to $t$ path in $G'$. If $P$ does not intersect $G_{\text{in}}$, we are done since this path is untouched by the algorithm so far, and faces no distortion. Otherwise, let $x_a$ and $x_b$ be the first and last vertices on $L_k$ when traveling from $s$ to $t$ via~$P$. 
    By the previous argument, we know that $d_{G'}(s, x_a) = d_G(s, x_a)$ and $d_{G'}(x_b, t) = d_G(x_b, t)$. Moreover, from~\Cref{lem:embed-outerplanar}, we have $d_G(x_a, x_b) \leq \E[d_{G'}(x_a, x_b)] \leq O(1) \cdot d_G(x_a, x_b)$.

    \medskip
    The above two cases imply that for every $(s, t)\in \mathfrak{T}$, we have:
    \[d_G(s, t)\leq \E[d_{G'}(s,t)]\leq O(1)\cdot d_G(s, t)\]
    The resulting graph $G'$ is $k$-outerplanar:~\Cref{lem:embed-outerplanar} claims that each $H_{in}^i$ is a face with some chords connecting its vertices. Moreover, the $\textsc{MakeStar}$ procedure converts the attached clusters into trees, so all these vertices also become part of this outer face. When attached with $G_{\text{out}}\setminus G[L_k]$, this component forms the final $k^{th}$ outerplanar layer of $G'$.

    Finally, we use~\Cref{thm:tree-embed} to probabilistically embed $G'$ into random tree $T^*$ with expected distortion $\poly(k)$. For each $(s, t)$, we have:
    \[\E[d_{T^*}(s, t) \mid G'] \leq \poly(k) \cdot d_{G'}(s, t)\]
    implying
    \[\E[d_{T^*}(s, t)] \leq \poly(k) \cdot \E[d_{G'}(s, t)]
        \leq \poly(k) \cdot d_G(s, t)\]
    The non-contraction follows similarly. This finishes the proof.
\end{proof}

\subsection*{AI Disclosure}
We used ChatGPT 5.5 and Gemini 3.1 for grammar, styling, and editing proofs for notational consistency.
%%%%%%%%%%%%%%%%%%%%%%%%%%%%%%%%%%%%%%%%%%%
%%%%%%%%%%%%%%%%%%%%%%%%%%%%%%%%%%%%%%%%%%%
\bibliographystyle{plain}
\bibliography{ref}

@article{CGNRS06,
author = {Chekuri, Chandra and Gupta, Anupam and Newman, Ilan and Rabinovich, Yuri and Sinclair, Alistair},
title = {Embedding $k$-Outerplanar Graphs into $\ell_1$},
journal = {SIAM Journal on Discrete Mathematics},
volume = {20},
number = {1},
pages = {119-136},
year = {2006},
doi = {10.1137/S0895480102417379},
note = {Announced at SODA '03}
}

@inproceedings{gnrs99,
author = {Gupta, Anupam and Sinclair, Alistair and Newman, Ilan and Rabinovich, Yuri},
title = {Cuts, Trees and $\ell_1$-Embeddings of Graphs},
year = {1999},
isbn = {0769504094},
publisher = {IEEE Computer Society},
address = {USA},
booktitle = {Proceedings of the 40th Annual Symposium on Foundations of Computer Science},
pages = {399},
series = {FOCS '99}
}

@article{GNRS04,
  title = {Cuts,  Trees and $\ell_1$-Embeddings of Graphs},
  volume = {24},
  DOI = {10.1007/s00493-004-0015-x},
  number = {2},
  journal = {Combinatorica},
  author = {Gupta,  Anupam and Newman,  Ilan and Rabinovich,  Yuri and Sinclair,  Alistair},
  year = {2004},
  pages = {233–269},
 note = {Announced at FOCS '99}
}

@inproceedings{Sidiropoulos2013,
  title={Non-positive curvature and the planar embedding conjecture},
  author={Sidiropoulos, Anastasios},
  booktitle={2013 IEEE 54th Annual Symposium on Foundations of Computer Science},
  pages={177--186},
  year={2013},
  series ={FOCS '13}
}

@article{OS81,
  title={Multicommodity flows in planar graphs},
  author={Okamura, Haruko and Seymour, Paul D},
  journal={Journal of Combinatorial Theory, Series B},
  volume={31},
  number={1},
  pages={75--81},
  year={1981}
}

@inproceedings{KLR19,
  title={Flow-cut gaps and face covers in planar graphs},
  author={Krauthgamer, Robert and Lee, James R and Rika, Havana},
  booktitle={Proceedings of the 30th Annual ACM-SIAM Symposium on Discrete Algorithms},
  pages={525--534},
  year={2019},
  series = {SODA '19}
}

@inproceedings{Filtser20,
  title={A face cover perspective to $\ell_1$ embeddings of planar graphs},
  author={Filtser, Arnold},
  booktitle={Proceedings of the 14th Annual ACM-SIAM Symposium on Discrete Algorithms},
  pages={1945--1954},
  year={2020},
  series = {SODA '20}
}

@article{Emek11,
  title={$k$-{O}uterplanar graphs, planar duality, and low stretch spanning trees},
  author={Emek, Yuval},
  journal={Algorithmica},
  volume={61},
  number={1},
  pages={141--160},
  year={2011},
  note = {Announced at ESA '09}
}

@article{LR09,
  title = {Coarse Differentiation and Multi-flows in Planar Graphs},
  volume = {43},
  DOI = {10.1007/s00454-009-9172-4},
  number = {2},
  journal = {Discrete \& Computational Geometry},
  author = {Lee,  James R. and Raghavendra,  Prasad},
  year = {2009},
  pages = {346–362}
}

@InProceedings{KK25,
  author =	{Kalantarzadeh, Sina and Kumar, Nikhil},
  title =	{Improved Lower Bounds on Multiflow-Multicut Gaps},
  booktitle =	{Approximation, Randomization, and Combinatorial Optimization. Algorithms and Techniques},
  pages =	{14:1--14:18},
  series =	{APPROX/RANDOM '20'25},
  year =	{2025},
  volume =	{353},
  doi =		{10.4230/LIPIcs.APPROX/RANDOM.2025.14},
note = {Arxiv version \url{https://arxiv.org/pdf/2507.06576}}
}

@inproceedings{CJLV08,
  title={Embeddings of topological graphs: Lossy invariants, linearization, and 2-sums},
  author={Chakrabarti, Amit and Jaffe, Alexander and Lee, James R and Vincent, Justin},
  booktitle={The 49th Annual IEEE Symposium on Foundations of Computer Science},
  pages={761--770},
  year={2008},
  series = {FOCS '08}
}

@article{AFGN22,
  title={Metric embedding via shortest path decompositions},
  author={Abraham, Ittai and Filtser, Arnold and Gupta, Anupam and Neiman, Ofer},
  journal={SIAM Journal on Computing},
  volume={51},
  number={2},
  pages={290--314},
  year={2022},
  publisher={SIAM}
}

@article{CSW13,
  title={Flow-cut gaps for integer and fractional multiflows},
  author={Chekuri, Chandra and Shepherd, F Bruce and Weibel, Christophe},
  journal={Journal of Combinatorial Theory, Series B},
  volume={103},
  number={2},
  pages={248--273},
  year={2013},
  note = {Announced at SODA '10}
}

@inproceedings{LS09,
  title={On the geometry of graphs with a forbidden minor},
  author={Lee, James R and Sidiropoulos, Anastasios},
  booktitle={Proceedings of the 41st annual ACM Symposium on Theory of computing},
  pages={245--254},
 year={2009},
  series = {STOC '09}
}

@article{Kumar25,
  title = {An Approximate Generalization of the Okamura–Seymour Theorem},
  volume = {54},
  DOI = {10.1137/23m1545938},
  number = {5},
  journal = {SIAM Journal on Computing},
  author = {Kumar,  Nikhil},
  year = {2025},
  pages = {FOCS22--159--FOCS22--177}
}

@inproceedings{LP13,
author = {Lee, James R. and Poore, Daniel E.},
title = {On the 2-sum embedding conjecture},
year = {2013},
doi = {10.1145/2462356.2492436},
booktitle = {Proceedings of the 29th Annual Symposium on Computational Geometry},
pages = {197–206},
numpages = {10},
series = {SoCG '13}
}

@inproceedings{CG04,
  author    = {Douglas E. Carroll and
               Ashish Goel},
  title     = {Lower Bounds for Embedding into Distributions over Excluded Minor
               Graph Families},
  booktitle = {Algorithms -- {ESA} 2004, 12th Annual European Symposium, Bergen, Norway,
               September 14-17, 2004, Proceedings},
  pages     = {146--156},
  year      = {2004},
  url       = {https://doi.org/10.1007/978-3-540-30140-0\_15},
  doi       = {10.1007/978-3-540-30140-0\_15}
  }

@InProceedings{Bartal96,
  Title     = {Probabilistic Approximations of Metric Spaces and Its Algorithmic
 Applications},
  Author    = {Yair Bartal},
  Booktitle = {37th Annual Symposium on Foundations of Computer Science, {FOCS} '96,
 Burlington, Vermont, USA, 14-16 October, 1996},
  Year      = {1996},
  Pages     = {184--193},
  DOI       = {10.1109/SFCS.1996.548477}
}

@inproceedings{Bartal98,
  author    = {Yair Bartal},
  title     = {On Approximating Arbitrary Metrices by Tree Metrics},
  booktitle = {Proceedings of the Thirtieth Annual {ACM} Symposium on the Theory
               of Computing, Dallas, Texas, USA, May 23-26, 1998},
  pages     = {161--168},
  year      = {1998},
  doi       = {10.1145/276698.276725}
  }

@article{FRT04,
  author    = {Jittat Fakcharoenphol and
               Satish Rao and
               Kunal Talwar},
  title     = {A tight bound on approximating arbitrary metrics by tree metrics},
  journal   = {J. Comput. Syst. Sci.},
  volume    = {69},
  number    = {3},
  pages     = {485--497},
  year      = {2004},
  url       = {https://doi.org/10.1016/j.jcss.2004.04.011},
  doi       = {10.1016/j.jcss.2004.04.011},
  bibsource = {dblp computer science bibliography, https://dblp.org}
}

@article{Bourgain85,
  title={On {L}ipschitz embedding of finite metric spaces in {H}ilbert space},
  author={Bourgain, Jean},
  journal={Israel Journal of Mathematics},
  volume={52},
  number={1},
  pages={46--52},
  year={1985}
}

@article{LLR95,
  title={The geometry of graphs and some of its algorithmic applications},
  author={Linial, Nathan and London, Eran and Rabinovich, Yuri},
  journal={Combinatorica},
  volume={15},
  number={2},
  pages={215--245},
  year={1995},
  publisher={Springer}
}

@article{AAABN06,
author = {Alon, Noga and Awerbuch, Baruch and Azar, Yossi and Buchbinder, Niv and Naor, Joseph (Seffi)},
title = {A general approach to online network optimization problems},
year = {2006},
volume = {2},
number = {4},
url = {https://doi.org/10.1145/1198513.1198522},
doi = {10.1145/1198513.1198522},
journal = {ACM Trans. Algorithms},
month = oct,
pages = {640–660},
numpages = {21}
}

@article{LR99,
  title={Multicommodity max-flow min-cut theorems and their use in designing approximation algorithms},
  author={Leighton, Tom and Rao, Satish},
  journal={Journal of the ACM},
  volume={46},
  number={6},
  pages={787--832},
  year={1999}
}

@article{AY98,
  title = {An ${O}(\log k)$ Approximate Min-Cut Max-Flow Theorem and Approximation Algorithm},
  volume = {27},
  DOI = {10.1137/s0097539794285983},
  number = {1},
  journal = {SIAM Journal on Computing},
  publisher = {Society for Industrial & Applied Mathematics (SIAM)},
  author = {Aumann,  Yonatan and Rabani,  Yuval},
  year = {1998},
  pages = {291–301}
}

@inproceedings{Rao99,
  title={Small distortion and volume preserving embeddings for planar and Euclidean metrics},
  author={Rao, Satish},
  booktitle={Proceedings of the 15th annual symposium on Computational geometry},
  pages={300--306},
  year={1999},
series = {SoCG '99}
}

@article{AGGNT19,
  title={Cops, robbers, and threatening skeletons: Padded decomposition for minor-free graphs},
  author={Abraham, Ittai and Gavoille, Cyril and Gupta, Anupam and Neiman, Ofer and Talwar, Kunal},
  journal={SIAM Journal on Computing},
  volume={48},
  number={3},
  pages={1120--1145},
  year={2019},
  publisher={SIAM}
}

@inproceedings{KLMN04,
  title={Measured descent: A new embedding method for finite metrics},
  author={Krauthgamer, Robert and Lee, James R and Mendel, Manor and Naor, Assaf},
  booktitle={45th Annual IEEE Symposium on Foundations of Computer Science},
  pages={434--443},
  year={2004},
  organization={IEEE}
}

@article{RS03,
  title={Graph minors. {XVI}. {E}xcluding a non-planar graph},
  author={Robertson, Neil and Seymour, Paul D},
  journal={Journal of Combinatorial Theory, Series B},
  volume={89},
  number={1},
  pages={43--76},
  year={2003},
  publisher={Elsevier}
}

@article{Baker94,
  title={Approximation algorithms for NP-complete problems on planar graphs},
  author={Baker, Brenda S},
  journal={Journal of the ACM},
  volume={41},
  number={1},
  pages={153--180},
  year={1994},
  publisher={ACM New York, NY, USA}
}

@InProceedings{Kummar20,
  author =	{Kumar, Nikhil},
  title =	{{Multicommodity Flows in Planar Graphs with Demands on Faces}},
  booktitle =	{31st International Symposium on Algorithms and Computation},
  pages =	{41:1--41:11},
  year =	{2020},
  doi =		{10.4230/LIPIcs.ISAAC.2020.41},
   series = {ISAAC '20}
}

@article{Dorn10,
  title = {Dynamic programming and planarity: Improved tree-decomposition based algorithms},
  volume = {158},
  DOI = {10.1016/j.dam.2009.10.011},
  number = {7},
  journal = {Discrete Applied Mathematics},
  author = {Dorn,  Frederic},
  year = {2010},
  month = apr,
  pages = {800–808}
}

@article{DPBF10,
author = {Dorn, Frederic and Penninkx, Eelko and Bodlaender, Hans L. and Fomin, Fedor V.},
title = {Efficient Exact Algorithms on Planar Graphs: Exploiting Sphere Cut Decompositions},
year = {2010},
volume = {58},
number = {3},
doi = {10.1007/s00453-009-9296-1},
journal = {Algorithmica},
month = nov,
pages = {790–810},
numpages = {21}
}

@inproceedings{Gutap01,
  title={Steiner points in tree metrics don't (really) help},
  author={Gupta, Anupam},
  booktitle={Proceedings of the 12nd annual ACM-SIAM Symposium on Discrete algorithms},
  pages={220--227},
  year={2001},
 series = {SODA '01}
}

@article{LS13,
  title={Pathwidth, trees, and random embeddings},
  author={Lee, James R and Sidiropoulos, Anastasios},
  journal={Combinatorica},
  volume={33},
  number={3},
  pages={349--374},
  year={2013}
}

@inproceedings{CCCLPP25,
  title={Embedding Planar Graphs into Graphs of Treewidth ${O}(\log^3 n)$},
  author={Chang, Hsien-Chih and Cohen-Addad, Vincent and Conroy, Jonathan and Le, Hung and Pilipczuk, Marcin and Pilipczuk, Micha{\l}},
  booktitle={Proceedings of the 2025 Annual ACM-SIAM Symposium on Discrete Algorithms (SODA)},
  pages={88--123},
  year={2025},
  organization={SIAM}
}

@inproceedings{chang2024shortcut,
  title={Shortcut partitions in minor-free graphs: Steiner point removal, distance oracles, tree covers, and more},
  author={Chang, Hsien-Chih and Conroy, Jonathan and Le, Hung and Milenkovi{\'c}, Lazar and Solomon, Shay and Than, Cuong},
  booktitle={Proceedings of the 2024 Annual ACM-SIAM Symposium on Discrete Algorithms (SODA)},
  pages={5300--5331},
  year={2024},
  organization={SIAM},
  note={\url{https://arxiv.org/abs/2308.00555v1}}
}

\end{document}